\documentclass[11pt]{article}
\usepackage{fullpage}
\usepackage{times}
\usepackage{color}
\usepackage{amsmath}
\usepackage{amssymb}
\usepackage{amsthm}
\usepackage{nicefrac}
\usepackage{amsthm}
\usepackage{hyperref,verbatim}

\usepackage[linesnumbered,boxed]{algorithm2e}
\usepackage{graphicx}
\usepackage{enumerate}
\usepackage{appendix} 

\usepackage{caption,wrapfig,bm}
\usepackage[usenames,dvipsnames,svgnames,table]{xcolor}
\usepackage{thm-restate}
\usepackage{mdframed}
\usepackage{enumitem}

\hypersetup{%
  breaklinks,%
  ocgcolorlinks, %
  colorlinks=true,%
  urlcolor=[rgb]{0.45,0.0,0.0},%
  linkcolor=[rgb]{0.0,0.45,0.0},%
  citecolor=[rgb]{0,0,0.45}%
}%

\def\compactify{\itemsep=0pt \topsep=0pt \partopsep=0pt \parsep=0pt}
\let\latexusecounter=\usecounter

{\begin{itemize}%
\setlength{\itemsep}{0pt}%
\setlength{\topsep}{0pt}%
\setlength{\partopsep}{0pt}%
\setlength{\parskip}{0pt}}%
{\end{itemize}}

\newcommand{\eps}{\epsilon}

\def\mE{\mathbb{E}}
\newcommand{\EE}[2]{\mE_{#1}\left[#2\right]}

\newcommand{\wsp}{weak spectral spanner}
\def\WSP{Weak Spectral Spanner}
\newcommand{\wspa}[1]{weak #1-spectral spanner}
\def\hu{{\hat{u}}}
\def\hv{{\hat{v}}}
\def\hw{{\hat{w}}}
\newtheorem{theorem}{Theorem}[section]

\newtheorem{claim}[theorem]{Claim}
\newtheorem{lemma}[theorem]{Lemma}
\newtheorem{corollary}[theorem]{Corollary}

\newtheorem{proposition}[theorem]{Proposition}

\newtheorem{definition}[theorem]{Definition}

\newtheorem{fact}[theorem]{Fact}
\newenvironment{proofof}[1]{{\vspace{3mm}\em Proof of #1.}}{\hfill
	\qed}

\DeclareMathOperator{\Tr}{tr}

\makeatletter
\newcommand{\setword}[2]{%
	\phantomsection
	#1\def\@currentlabel{\unexpanded{#1}}\label{#2}%
}
\makeatother
\usepackage{listings}

\lstdefinestyle{nonumbers}{numbers=none}
\lstdefinestyle{numbers}{numbers=left, numberstyle=\footnotesize, 
stepnumber=1, numbersep=7pt, 
xleftmargin=10pt}
\def\Re{{\mathbb{R}}}

\def\Pr{{\mathbb{P}}}

\def\cS{{\cal {S}}}

\def\sI{{\mathcal{I}}}

\def\abs#1{\left|#1\right|}
\def\norm#1{\|#1\|}
\def\dist{\mathrm{dist}}

\def\sC{{\mathcal{C}}}

\DeclareMathOperator{\opt}{OPT}
\DeclareMathOperator{\argmax}{argmax}
\DeclareMathOperator{\argmin}{argmin}
\DeclareMathOperator{\rank}{rank}

\title{Composable Core-sets for Determinant Maximization Problems \\via Spectral Spanners}
\author{ \begin{tabular}{ccccc}
{\begin{tabular}{c} Piotr Indyk \\ MIT \\\small{indyk@mit.edu} \end{tabular}} & & & &
{\begin{tabular}{c}Sepideh Mahabadi \\TTIC \\\small{mahabadi@ttic.edu} \end{tabular}}  \tabularnewline \\
{\begin{tabular}{c}Shayan Oveis Gharan \\University of Washington \\\small{shayan@cs.washington.edu } \end{tabular}} & & & &
 {\begin{tabular}{c}Alireza Rezaei \\University of Washington \\\small{arezaei@cs.washington.edu} \end{tabular}} 
\end{tabular}
}
\date{}
\begin{document}

\maketitle

\begin{abstract}
We study a  generalization of classical combinatorial graph spanners to the spectral setting. Given a set of vectors $V\subseteq \Re^d$, we say a set $U\subseteq V$ is an $\alpha$-spectral $k$-spanner, for $k\leq d$, if for all $v\in V$ there is a probability distribution $\mu_v$ supported on $U$ such that 
$$vv^\intercal \preceq_k \alpha\cdot\mE_{u\sim\mu_v} uu^\intercal,$$ 
where for two matrices $A,B\in\Re^{d\times d}$ we write $A\preceq_k B$ iff the sum of the bottom $d-k+1$ eigenvalues of $B-A$ is nonnegative. In particular, $A\preceq_d B$ iff $A\preceq B$.
We show that any set $V$ has an $\tilde{O}(k)$-spectral spanner of size $\tilde{O}(k)$ and this bound is almost optimal in the worst case. 

We use spectral spanners to study composable core-sets for spectral problems. We show that for many objective functions one can use a spectral spanner, independent of the underlying function, as a core-set and obtain almost optimal composable core-sets.
For example, for the $k$-determinant maximization problem, we obtain an $\tilde{O}(k)^k$-composable core-set, and we show that this is almost optimal in the worst case. 

Our algorithm is a spectral analogue of the classical greedy algorithm for finding (combinatorial) spanners in graphs. We expect that our spanners find many other applications in distributed or parallel models of computation.  Our proof is  spectral. As a side result of our techniques, we show that the rank of diagonally dominant lower-triangular matrices are robust under ``small perturbations'' which could be of independent interests.

\end{abstract}

\thispagestyle{empty}
\setcounter{page}{0}
\newpage

\section{Introduction}
Given a graph $G$ with $n$ vertices $\{1,\dots,n\}$, we say a subgraph $H$ is a $\alpha$-(combinatorial) spanner if for every pair of vertices $u,v$ of $G$,
$$\dist_H(u,v) \leq \alpha\cdot\dist_G(u,v),$$
where $\dist_G(u,v)$ is the shortest path distance between $u,v$ in $G$. It has been shown that for any $\alpha$, $G$ has an $\alpha$-spanner with only $n^{1+O(1)/\alpha}$ many edges and that can be found efficiently \cite{elkin2004}. Such a spanner can be found by a simple algorithm which repeatedly finds and adds an edge $f=(u,v)$ where $\dist_H(u,v) > \alpha$.
Combinatorial spanners have many applications in distributed computing \cite{peleg2000distributed,francis2001idmaps,korkmaz2000source}, optimization \cite{dodis1999design,arora1998polynomial}, etc.
 
In this paper we define and study a spectral generalization of this property.
Given a set of vectors $V\subseteq \Re^d$, we say a set $U\subseteq V$ is an $\alpha$-spectral $d$-spanner of $V$ if for any vector $v\in V$, there exists a probability distribution $\mu_v$ on the vectors in $U$ such that
$$ vv^\intercal \preceq \alpha \cdot \mE_{u\sim \mu_v}uu^\intercal \hspace{0.5cm} \text{equiv} \hspace{0.5cm} \langle x,v\rangle^2 \leq \alpha\cdot\mE_{u\sim\mu_v} \langle x,u\rangle^2, \forall x\in\Re^d.$$ 
To see that this is a generalization of the graph case, let $b_{u,v}=e_u-e_v$ be the vector corresponding to an edge $\{u,v\}$ of $G$, where $e_u$ is the indicator vector of the vertex $u$. It is an exercise to show that for $V=\{b_e\}_{e\in E(G)}$ and for any $\alpha$-combinatorial spanner $H$ of $G$, the set $U=\{b_e\}_{e\in E(H)}$ is an $\alpha^2$-spectral spanner of $V$.

The following theorem is a special case of our main theorem.
\begin{theorem}[Main theorem for $k=d$]\label{thm:main-spectralspanner}
There is an algorithm that for any set of vectors $V\subseteq \Re^d$ finds an $\tilde{O}(d)$-spectral $d$-spanner of size $\tilde{O}(d)$ in time polynomial in $d$ and size of $|V|$ \footnote{The asymptotic notation $\tilde{O}(f(n))$ hides logarithmic factors in $f(n)$.}. 
\end{theorem}
Our algorithm is a spectral generalization of the greedy algorithm mentioned above for finding combinatorial spanners.

We further study generalizations of our spectral spanners to weaker forms of PSD inequalities. 
For two matrices $A,B$ we write $A\preceq_k B$ if for every projection matrix $\Pi$ onto a $d-k+1$ dimensional linear subspace, $\langle A,\Pi\rangle \leq \langle B,\Pi\rangle$.
For example, if $A\preceq_k B$, then sum of the top $k$ eigenvalues of $A$ is at most the sum of the top $k$ eigenvalues of $B$. Analogously, we say $U\subseteq V$ is an $\alpha$-spectral $k$-spanner of $V$, if for any $v\in V$, there is a distribution $\mu_v$ on $U$ such that $vv^\intercal \preceq_k \alpha\cdot \mathbb{E}_{u\sim\mu_v} uu^\intercal.$ In our main theorem we generalize the above statement to all $k\leq d$ and we show that to construct an $\tilde{O}(k)$ spectral $k$-spanner we only need to use $\tilde{O}(k)$ many vectors independent of the ambient dimension of the space. 

\begin{theorem}[Main]\label{thm:mainspannerfork}
	There is an algorithm that for any set of vectors $V\subseteq \Re^d$ finds an $\tilde{O}(k)$-spectral $k$-spanner of size $\tilde{O}(k)$.
	
	Furthermore, for any $\eps>0$ and $k\leq d$, there exists a set $V\subseteq \Re^d$ of size $e^{\Omega(k^\eps)}$ such that any $k^{1-\eps}$-spectral spanner of $V$ must have all vectors of $V$.
\end{theorem}

\paragraph{Composable core-sets.} Our main application of spectral spanners is to design (composable) core-sets for spectral problems. 
 A function $c(V)$ that maps $V \subseteq \Re^d$ into its subset is called an $\alpha$-composable core-set of size $t$ for the function $f(\cdot)$ ~\cite{AHV-gavc-05,IMMM-ccdcm-14}, if for any collection of  sets $V_1, \ldots, V_p \subset  \Re^d$, we have 
\[ f(c(V_1) \cup \ldots \cup c(V_p)) \ge \frac1{\alpha} \cdot f(V_1\cup \ldots \cup V_p) \]
and $|c(V_i)| \le t$ for any $V_i$.  A composable core-set of a small size immediately yields a communication-efficient distributed approximation algorithm: if each set $V_i$ is stored on a separate machine, then all machines can compute and transmit core-sets, $c(V_i)$'s, to a central server, which can then perform the final computation over the union. Similarly,  core-sets make it possible to design a streaming algorithm which processes $N$ vectors in one pass using only $\sqrt{Nt}$ storage. This is achieved by dividing the stream of data into blocks of size $\sqrt{Nt}$, computing and storing a core-set for each block, and then performing the computation over the union.

In this paper we show that, for a given set $V_i\in\Re^d$, an $\alpha$-spectral spanner of $V_i$  for a proper value of $\alpha$ provides a good core-set of $V_i$s. Specifically, we show that for many (spectral) optimization problems,  such as determinant maximization, $D$-optimal design or min-eigenvalue maximization,  this approach leads to almost the best possible composable core-set in the worst case. 

In what follows we discuss a specific application, to determinant maximization, in more detail. 

\subsection{Composable Core-sets for Determinant Maximization}
Determinantal point processes are widely popular probabilistic models. Given a set of vectors $V\subseteq \Re^d$, and a parameter $k$, DPPs sample $k$-subsets $S$ of $P$ with probability $\Pr(S)$ proportional to the squared volume of the  parallelepiped spanned by the elements  of $S$. That is,
\[ \Pr(S) \sim  \det_k(\sum_{v \in S} vv^T) \]
This distribution formalizes a notion of diversity, as sets of vectors that are ``substantially different'' from each other are assigned higher probability. One can then find the ``most diverse" $k$-subset in $P$  by computing  $S$ that maximizes $\Pr(S)$, i.e., solving the {\em maximum a posteriori (MAP) decoding} problem:
\[ \max_{S \subset P, |S|=k} \Pr(S).  \]
We also refer to this problem as the \emph{$k$-determinant maximization} problem. Since their introduction to machine learning  at the beginning of this decade~\cite{kulesza2010structured,kulesza2011k,kulesza2012determinantal},  DPPs have found applications to 
video summarization~\cite{mirzasoleiman2017streaming,gong2014diverse}, 
document summarization~\cite{kulesza2012determinantal, chao2015large, kulesza2011learning}, 
tweet timeline generation \cite{yao2016tweet},
object detection~\cite{lee2016individualness},
and other applications. All of the aforementioned applications involve MAP decoding.

Here we use our results on spectral spanners to construct an almost optimal composable core-set for MAP problem. Before mentioning our result let us briefly discuss relevant previous work on this problem.
The MAP problem is hard to approximate up to a factor of $2^{-ck}$ for some constant $c>0$, unless P$=$NP. \cite{cm-smvsm-09, civril2013exponential}. This lower bound was matched qualitatively by a recent paper of~\cite{nikolov2015randomized}, who gave an algorithm with  $e^k$-approximation guarantee.
Since the data sets in the aforementioned applications can be large, there has been a considerable effort on developing efficient algorithms in distributed, streaming or parallel models of computation~\cite{mirzasoleiman2017streaming,wei2014fast,pan2014parallel,mirzasoleiman2013distributed,mirzasoleiman2015distributed,mirrokni2015randomized,barbosa2015power}.  All of these algorithms relied on the fact that  the logarithm of the volume is a submodular function, which makes it possible to obtain multiplicative factor approximation algorithms (assuming some lower bound on the volume, as otherwise the logarithm of the volume can be negative). See Section~\ref{ss:related} for an overview. 
However, this generality comes at a price, as multiplicative approximation guarantees for the logarithm of the volume translates into "exponential" guarantees for the volume, and necessitates the aforementioned extra lower bound assumptions.  As a result, to the best of our knowledge, no multiplicative approximation factor algorithms were known before for this problem,  for streaming, distributed or parallel models of computation. 


 In this paper we present 
the first (composable) core-set construction for the determinant maximization problem. Our main contributions are:
\begin{restatable}{theorem}{mainthmupperbound}
	\label{thm:main_upper_bound}
	There exists a polynomial time algorithm for computing an $\tilde O(k)^k$-composable core-set of size $\tilde O(k)$, for the $k$-determinant maximization problem.
\end{restatable}
Let us discuss the proof of the above theorem for the case $k=d$ using our main theorem \ref{thm:main-spectralspanner}.
Given sets of vectors $V_1,\dots, V_m$ let $U_1,\dots,U_m$ be their $\tilde O(d)$-spectral spanners respectively. Let 
$$S=\argmax_{S\subseteq \cup_i V_i: |S|=d} \Pr(S).$$
Consider the matrix  $A=\sum_{v\in S}  \mE_{u\sim\mu_v} uu^T$, that is we substitute each vector $v$ in $S$ by a convex combination of the vectors in the spectral spanner(s). Then, by definition of spectral spanner,
$$ \frac1{\alpha} \sum_{v\in S} vv^T \preceq A.$$
Since determinant is a monotone function with respect to the Loewner order of PSD matrices,
$$ \frac1{\alpha^d} \Pr(S) =\det\left(\frac1\alpha\sum_{v\in S} vv^T\right) \leq \det(A). $$
The matrix $A$ can be seen as a fractional solution to the determinant maximization problem. In fact \cite{nikolov2015randomized} showed that $A$ can be rounded to a set $T$ of size $|T|=d$ such that $\det(A)\leq e^d \det(T)$. Therefore, we obtain an $(e\alpha)^d$ approximation for determinant maximization (see \autoref{sec:detmaximization} for more details).

The technique that we discussed above can be applied to many optimization problems. In general, if instead of the determinant, we wanted to maximize any function $f:\Re^{d\times d}\to\Re_+$, that is monotone on the Loewner order of PSD matrices, we can use the above approach to construct a fractional solution $A$ supported on the spectral spanners such that $f(A)$ is at least the optimum (up to a loss that depends on $\alpha$). Then, we can use randomized rounding ideas to round the the matrix $A$ to an integral solution of $f$. See \autoref{sec:experimentaldesign} for further examples.

We complement the above theorem by showing the above guarantee is essentially the best possible.
\begin{theorem}\label{thm:main_lower_bound}
Any composable core-sets of size at most  $k^{\beta}$ for the $k$-determinant maximization problem must incur an approximation factor of at least $(\frac{k}{\beta})^{k(1-o(1))}$, for any $\beta \ge 1$.
\end{theorem}

Note that our lower bound of $(\frac{k}{\beta})^{k(1-o(1))}$  for the approximation factor achievable by composable core-sets is substantially higher than the approximation factor  $e^k$ of the best off-line algorithm, demonstrating a large gap between these two models.


\subsection{Overview of the Techniques}\label{sec:techniques}
In this part, we give a high level overview of the proof
of \autoref{thm:mainspannerfork}.
Our proof has two steps: First, we solve the ``full dimensional version of the problem, i.e., we construct an $\tilde{O}(d)$-spectral $d$-spanner of size $\tilde{O}(d)$ for a given set of vectors in $\Re^d$ as promised in \autoref{thm:main-spectralspanner}. Then, we reduce the ``low dimensional'' version of the problem, i.e., finding $k$-spanners for $k<d$, to the full dimensional version in a $\tilde{O}(k)$-dimensional space.

\paragraph{Step 1: } Our high-level plan is  to  ``augment'' the classical greedy algorithm for finding combinatorial spanners in graphs to the spectral setting.
First, we rewrite the combinatorial algorithm in spectral language.

Let $G$ be  a graph with vertex set $V(G)$ and edge set $E(G)$. Recall that  for any edge $e=\{u,v\} \in E(G)$  $b_e =e_u-e_v$.  As alluded to in the introduction, if $H$ is an $\alpha$-combinatorial spanner of $G$, then $U=\{b_e\}_{e\in E(H)}$ is an $\alpha^2$-spectral spanner  of $\{b_e\}_{e\in E(G)}$.
The following algorithm gives an $\alpha$-combinatorial spanner with $n^{1+O(1)/\alpha}$ edges: Start with an empty graph $H$. While there is an edge $f=\{u,v\}$ in $G$ where $\dist_H(u,v)>\alpha$, add it to $H$. One can observe that $\dist_H(u,v)>\alpha$ iff, for any distribution $\mu$ on $E(H)$, $b_fb_f^T \not\preceq\alpha^2 \mE_{e\sim \mu} b_eb_e^T$.
 


This observation suggests a natural algorithm  in the spectral setting: 
At each step find a vector $v\in V$ such that for all $\mu$ supported on the set of vectors already chosen in the spanner, $vv^T\not\preceq \alpha\cdot \mE_{u\sim \mu} uu^T$, and add it to the spanner. We can implement such an algorithm in polynomial time, but we cannot directly bound the size of the spectral spanner that such an algorithm constructs using our current techniques.

So, we take a detour. First, we solve a seemingly easier problem by changing the 
order of quantifiers in the definition of the spectral spanner. 
For $V\subseteq \Re^d$, a subset $U \subseteq V$ is a \wspa{$\alpha$}\ of $V$, if for all $v\in V$ and $x\in \Re^d$ there is a distribution $\mu_{v,x}$ on $U$ such that
$$ \langle v,x\rangle^2 \leq \alpha\cdot \mE_{u\sim\mu_{v,x}} \langle u,x\rangle^2 \hspace{0.5cm} \text{equiv} \hspace{0.5cm}
 \langle v,x\rangle^2 \leq \alpha \cdot \max_{u\in U}\langle u,x\rangle^2.$$
To find a \wsp, we use the analogue of the greedy algorithm: Let $U$ be the set of vectors already chosen; while there is a vector $v\in V$ and $x\in \Re^d$ such that $\langle x,v\rangle^2 > \alpha\cdot\max_{u\in U} \langle u,x\rangle^2$ we add $\argmax_{v} \langle x,v\rangle^2$ to $U$.

We prove that for $\alpha=\tilde{O}(d)$ the above algorithm stops in $\tilde{O}(d)$ steps. Suppose  that the algorithm finds vectors $u_1,\dots,u_m$ together with corresponding ``bad'' directions $x_1,\dots,x_m$, where $x_i$ being a bad direction for $u_i$ means that \begin{equation}\label{eq:tunealpha}\langle u_i,x_i\rangle^2 >\alpha \langle u_j,x_i\rangle^2, \forall 1\leq i\leq m, \forall 1\leq j<i.	
 \end{equation}
We need to show that $m=\tilde{O}(d)$. We consider the matrix $M\in\Re^{m\times m}$ where $M_{i,j}=\langle u_i,x_j\rangle$. By the above constraints $M$ is diagonally dominant and approximately lower triangular matrix.  But since $M$ has rank at most $d$ as it is the inner product matrix of vectors lying in $\Re^d$, we conclude that $m=\tilde{O}(d)$. 
Note that in the extreme case, where $M$ is truly lower triangular the latter fact obviously holds because then $\rank(M)=m$. 
As a side result, we also show that the rank of lower triangular matrices is robust under small perturbations, (see  \autoref{lem:approxlowtri}). 

The above argument shows that the spectral greedy algorithm gives a \wsp \  for $\alpha=\tilde{O}(d)$ of size $\tilde{O}(d)$. To finish the proof of \autoref{thm:main-spectralspanner} we need to find a (strong) $\alpha$-spectral spanner from our weak spanner. We use a duality argument to show that any \wsp\ is indeed an $\alpha$-spectral spanner. Let $U$ be a \wsp. To verify that $U$ is an $\alpha$-spectral spanner, we need to find a distribution $\mu_v$ for any $v\in V$ supported on $U$ such that $vv^T\preceq \alpha\cdot\mE_{u\sim\mu_v} uu^T$.  We can find the best distribution $\mu_v$ using an SDP with variables $p_u$ for all $u\in U$ denoting $\Pr_{\mu_v}(u)$. Instead of directly bounding the primal, we write down the dual of the SDP and use hyperplane separating theorem to show that indeed such a distribution exists.

It was pointed to us by an anonymous reviewer that one can use approximate John's ellipsoid \cite{ball1997elementary} to find an $O(d)$-weak-spectral $d$-spanner of size $\tilde{O}(d)$. This improves the guarantees of our algorithm by a $\log d$ factor. We discuss the details  at the end of Section \ref{subsection:analdirheight}. Let us briefly mention the advantages of our algorithm over the John's ellipsoid method: First, in finding the weak spanner one can tune the value of $\alpha$ in \eqref{eq:tunealpha} based on the structure of the given data points and the ideal size of the core-set. We also expect that in many real world applications, one can use our algorithm to obtain $\text{polylog}(d)$-spectral $d$-spanners of size $\tilde{O}(d)$. Secondly, to implement our algorithm we only need to solve linear programs with $O(d)$ many variables. This requires polynomially smaller amount of memory compared to the SDP solvers one needs to use to solve the John's ellipsoid. Lastly, our algorithm is easier to parallelize.

\paragraph{Step 2: }
To reduce the $k$-spanner problem to the ``full dimensional'' case, we use the greedy algorithm of \cite{cm-smvsm-09}  to 
 find a set of vectors $W \subset V$ of size  $\tilde{O}(k)$ such that for any $v \in V$,
 \begin{equation}
     \label{eq:kspanfirst}
     v_{\langle W \rangle^\perp}v_{\langle W \rangle^\perp}^\intercal \preceq_k O(1)\cdot \EE{w}{ww^\intercal}
 \end{equation}
 where  the expectation is over the uniform distribution on $W$, and  $v_{\langle W \rangle^\perp}$ represents the projection of $v$ onto the space orthogonal to the linear subspace spanned by $W$.
Then, we project all vectors in $V$ onto the space $\langle W \rangle$, and we solve the full dimensional version, i.e.,  we find $U \subseteq V$ of size $\tilde{O}(|W|)$ such that for any $v \in V$, there exists a distribution $\mu_v$ supported on $U$ which satisfies
\begin{equation}
\label{eq:kspansec}
    v_{\langle W \rangle} v_{\langle W \rangle}^\intercal \preceq \tilde{O}(k)\cdot \EE{u \sim \mu_v}{u_{\langle W\rangle} u_{\langle W \rangle}^\intercal}.
\end{equation}
Ideally on the RHS of the above, we need to have $uu^\intercal$ instead of $u_{\langle W\rangle} u_{\langle W \rangle}^\intercal$ which can be achieved  by incurring an extra constant factor by applying \eqref{eq:kspanfirst}. It is not hard to see from the above two equations that $U\cup W$ is an $\tilde{O}(k)$-spectral $k$-spanner for $V$.

It remains to find the set $W\subseteq V$ satisfying \eqref{eq:kspanfirst}.
We use the following algorithm: Let 
$W=\emptyset$. For $i=1,\dots,\tilde{O}(k)$, add $\argmax_{v\in V} 
\norm{v_{\langle U \rangle^\perp}}$ to $W$. Intuitively, we greedily choose a set of vectors of size $\tilde{O}(k)$ to minimize the projection of the remaining vectors in the orthogonal space of $\langle W\rangle$. 

To prove \eqref{eq:kspanfirst}, we need to show 
$$v_{\langle W \rangle^\perp} v_{\langle W\rangle 
^\perp}^\intercal \preceq_k O(1) \cdot 
\EE{w\sim\mu}{ ww^\intercal}.$$
Equivalently, after choosing the worst projection matrix onto a $d-k+1$ linear subspace, it is enough to show 
\begin{equation}
\label{eq:greedyvolgoal}
\norm{v_{\langle W \rangle^\perp}}^2 \leq O(1) \cdot \sum_{i=k}^d \lambda_i(\EE{w \sim \mu}{ww^\intercal}).
\end{equation}
 To prove the above inequality, we use properties of the greedy algorithm to study singular values of the matrix obtained by applying the Gram-Schmidt process on the vectors in $W$. 

\paragraph{Lower bounds.} 
As we discussed in the intro, it is not hard to prove that the guarantee of \autoref{thm:main-spectralspanner} is tight in the worst case. However, one might wonder if it is possible to design better composable core-sets for determinant maximization and related spectral problems. We show that for many such problems we obtain the best possible composable core-set in the worst case. Let us discuss the main ideas of \autoref{thm:main_lower_bound}.

We consider the case $k=d$ for simplicity of exposition. 
For a set $V\subseteq \Re^{d}$ and a linear transformation $Q\in\Re^{d\times d}$, define $QV=\{Qv\}_{v\in V}$.
Choose  a set $V\subseteq \Re^d$ of unit vectors such that for any distinct $u,v \in V$, $ \langle u,v \rangle^2  \le 1/d^{1-o(1)}$. This can be achieved with high probability by selecting points in $V$ independently and uniformly at random from the unit sphere. Recall that the set $V$ can have exponentially (in $d$) large number of vectors.
Consider sets $A_1, \dots, A_{d}$ and $B_1, \dots, B_{d-1}$ in a $(2d-1)$-dimensional space such that: 
\begin{itemize}
\item For each $1\leq i\leq d$, let $A_i=R_iV$ where $R_i$ is a rotation matrix which maps
$\Re^d$ to $\langle e_1,e_2, \dots e_{d-1}, e_{(d-1)+i}\rangle $ and it maps a uniformly randomly chosen vector of $V$ to $e_{(d-1)+i}$.
\item For each $1\leq i\leq d-1$, $B_i = \{M e_i\}$, where $M$ is a ``large'' number.
\end{itemize}
Our instance of determinant maximization is  simply $QA_1,\dots,QA_d, QB_1,\dots,QB_{d-1}$ for a random rotation matrix $Q\in\Re^{(2d-1)\times(2d-1)}$. 

Observe that the optimal set of $2d-1$ vectors contains $Qe_{(d-1)+i}$'s from $QA_i$'s and $QMe_{i}$'s from $B_i$'s, and has value equal to $(M^{d-1})^2$. However, since $Q$ is a random rotation, the core-set function cannot determine which vector in $QA_i$  was aligned with $Qe_{(d-1)+i}$. Recall that the core-set algorithm must find a core-set of $A_i$ by only observing the vectors in $A_i$.  Thus unless core-sets are exponentially large in $d$, there is a good probability  that, for all $i$,  the core-set for  $QA_i$ does not contain $Qe_{(d-1)+i}$. For a sufficiently large $M$, all vectors $QMe_i$ from $QB_i$ must be included in any solution with a non-trivial approximation factor. 
It follows that, with a constant probability,  any core-set-induced solution is sub-optimal by at least a factor of $d^{\Theta(d)}$.

\paragraph{Paper organization.} 
In \autoref{sec:spectralspannerdef}, we formally define spectral $(d)$-spanners and their generalization ``spectral $k$-spanners''. 
In \autoref{sec:spectralspannerford}, we introduce our  
algorithm for finding spectral spanners and prove 
\autoref{thm:main-spectralspanner}. Then, in 
\autoref{sec:spannerk}, we generalize the result of \autoref{thm:main-spectralspanner} for spectral $k$-spanners  and show the reduction from $k<d$ to the full dimensional case of $k=d$, proving 
\autoref{thm:mainspannerfork}. We mention applications of spectral spanners for designing composable core-sets for several optimization problems including $k$-determinant maximization in \autoref{sec:app}. In particular we prove \autoref{thm:main_upper_bound}. 
Finally in \autoref{sec:lower_bound}, we present our lower bound results and prove \autoref{thm:main_lower_bound}.

\subsection{Related work}
\label{ss:related}
As mentioned earlier, multiple papers developed composable core-sets (or similar constructions) when the objective function is equal to the {\em logarithm} of the volume. In particular, \cite{mirzasoleiman2013distributed} showed that core-sets obtained via a greedy algorithm guarantee an  approximation factor of $\min(k,n)$. The approximation ratio can be further improved to a constant if the input points are assigned to set $V_i$ uniformly at random \cite{mirrokni2015randomized,barbosa2015power}. However, these guarantees do not translate into a multiplicative approximation factor for the volume objective function.\footnote{It is possible to show that the greedy method achieves composable core-sets with multiplicative approximation factor of $2^{O(k^2)}$. Since this bound is substantially larger than our bound  obtained by spectral spanners, we do not include the proof in this paper.}

Core-sets constructions are known for a wide variety of geometric and metric problems, and several algorithms have found practical applications.  Some of those constructions are relevant in our context. In particular, core-sets for approximating the directional width \cite{AHV-aemp-04} have functionality that is similar to weak spanners. However, the aforementioned paper considered this problem for low-dimensional points, and as a result, the core-sets size was exponential in the dimension. Another line of research~\cite{abbar2013diverse,IMMM-ccdcm-14,ceccarello2017mapreduce} considered core-sets for maximizing  {\em metric} diversity measures, such as the minimum inter-point distance. Those measures rely only on relationships between pairs of points, and thus have quite different properties from the volume-induced measure considered in this paper. 

We also remark that one can consider generalizations of our problem to settings were we want to maximize the volume under additional constraints. Over the last few years several extensions were studied extensively and many new algorithmic ideas were developed \cite{NS16,AO17,SV17,ESV17}. In this paper, we study composable core-sets for the basic version of the determinant maximization problem where no additional constraints are present.

\section{Preliminaries}
\subsection{Linear Algebra}
\label{subsection:linalg}
Throughout the paper, all vectors that we consider are column based and sitting in $\Re^d$, unless otherwise 
specified. For a vector $v$, we use notation $v(i)$ to denote its $i_{\text{th}}$ coordinate and use $\norm{v}$ to denote its $\ell_2$ norm. Vectors $v_1,\ldots,v_k$ are called orthonormal if for any $i$, $\norm{v_i}=1$,  
and for any $i\neq j$, $\langle v_i,v_j\rangle=0$. For a set of vectors $V$, we let $\langle V\rangle$ denote the linear subspace spanned by vectors of $V$. We also use $S^\perp$ to denote the linear subspace orthogonal to $S$, for a linear subspace $S$.

Notation $\langle,\rangle$ is used to denote Frobenius inner 
product of matrices, for matrices $A,B\in \Re^{d\times d}$ 
 $$\langle A,B\rangle = \sum_{i=1}^d\sum_{j=1}^d A_{i,j}B_{i,j}=\Tr(AB^\intercal)$$ 
where $A_{i,j}$ denotes the entry of matrix $A$ in row $i$ and column $j$. 
Matrix $A$ is a Positive Semi-definite (PSD) matrix denoted by $A \succeq 0$ if 
it is symmetric and  for any vector $v$, we have $v^\intercal Av = \langle A,vv^\intercal 
\rangle \geq 0$. For PSD matrices $A,B$ we write
$A\preceq B$ if $B-A \succeq 0$. We also denote the set of $d\times d$ PSD matrices by $\mathcal{S}_d^+$.
\paragraph{Projection Matrices.} A matrix $\Pi \in \Re^{d\times d}$ is a projection matrix if  $\Pi^2=\Pi$.  It is also easy to see that for any $v\in \Re^d$, 
$$\langle vv^\intercal, \Pi \rangle = v^\intercal \Pi v= \langle \Pi v, \Pi v \rangle =\|\Pi v\|^2.$$
For a linear subspace $S$, we let $\Pi_S$ denote the matrix  projecting vectors from $\Re^d$ onto $S$ which means for any vector $v$, $(\Pi_S)v$ is the projection of $v$ onto $S$. If $S$ is $k$-dimensional and $v_1,\ldots,v_k$ form
an arbitrary orthonormal basis of $S$, then one can see that  $\Pi_S=\left( \sum_{i=1}^k 
v_iv_i^\intercal \right)$. We also represent the set of all projection matrices onto  $k$-dimensional subspaces  by $\mathcal{P}_k$.

\begin{fact}\label{fact:PiCauchy}
For any vectors $u,v\in\Re^d$ and any projection matrix $\Pi\in\Re^{d\times d}$
$$ \langle (u+v)(u+v)^\intercal,\Pi\rangle \leq 2 \langle uu^\intercal +vv^\intercal, \Pi\rangle.$$
\end{fact}
\begin{proof}
Let $a=\Pi u$ and $b=\Pi v$. Then since $\Pi^2=\Pi$, we have $\langle uu^\intercal,\Pi \rangle = \norm{a}^2$, $\langle vv^\intercal,\Pi \rangle = \norm{b}^2$, and $\langle (u+v)(u+v)^\intercal,\Pi \rangle =\norm{a+b}^2$. Now, the assertion is equivalent to
$$ \norm{a+b}^2 \leq 2(\norm{a}^2+\norm{b}^2)$$
which follows by Cauchy-Schwarz inequality,
$\langle a,b\rangle \leq \norm{a}\norm{b} \leq (\norm{a}^2+\norm{b}^2)/2$.
\end{proof}

\paragraph{Eigenvalues and Singular Values.} For a symmetric Matrix $A$, $\lambda_1(A)\geq \ldots \geq \lambda_d(A)$  denotes the eigenvalues of $A$. The following characterization of eigenvalues of symmetric matrices known as min-max or variational theorem is useful for us. 
\begin{theorem}[Min-max Characterization of Eigenvalues]
Let $A\in \Re^{d\times d}$ be a symmetric matrix with eigenvalues $\lambda_1 \geq \lambda_2 \geq \ldots\geq \lambda_d$. Then 
$$\lambda_k = \max \{ \min_{x \in U} \frac{x^\intercal Ax}{\norm{x}^2} \hspace{2mm} | \hspace{2mm} \,  U \text{ is a $k$-dimensional linear subspace}\},$$
or
$$\lambda_k = \min \{ \max_{x \in U} \frac{x^\intercal Ax}{\norm{x}^2} \hspace{2mm} | \hspace{2mm} \,  U \text{ is a $(d-k+1)$-dimensional linear subspace}\},$$
\end{theorem}
The following theorem known as Cauchy interlacing theorem shows the relation between eigenvalues of a symmetric matrix and eigenvalues of its submatrices.
\begin{theorem}[Cauchy Interlacing Theorem]
Let $A$ be a symmetric $d\times d$ matrix, and $B$ be an $m \times m$ principal submatrix of $A$. Then for any $1 \leq i \leq m$, we have 
$$ \lambda_{d-m+i}(A) \leq \lambda_i(B) \leq \lambda_i(A)$$
\end{theorem}
We also use the following lemma which is an easy implication of min-max characterization  to bound eigenvalues of summation of two matrices in terms of the summation their eigenvalues. 
\begin{lemma}
\label{lem:boundingeigvals}
Let $A,B \in \Re^{d\times d}$ be two symmetric matrices. Then $\lambda_{i+j-d}(A+B) \geq \lambda_i(A)+\lambda_j(B)$ for any $i,j$ with $i+j-d >0$.
\end{lemma}
\begin{proof}
Following the min-max characterization of eigenvalues, let $S_A$ and $S_B$ be two $i$-dimensional and $j$-dimensional linear subspaces for which we have 
\begin{equation*}
\begin{array}{ccc}
\lambda_i(A) = \min_{x\in S_A} \frac{x^\intercal Ax}{\norm{x}^2}& \text{and} & \lambda_j(B) = \min_{x\in S_B} \frac{x^\intercal Bx}{\norm{x}^2}  
\end{array}
\end{equation*}
Then let $S=S_A \cap S_B$. The dimension of $S$ is at least $i+j-d$, and by min-max characterization of eigenvalues we have 
$$ \lambda_{i+j-d}(A+B) \geq \min_{x\in S} \frac{x^\intercal(A+B)x}{\norm{x}^2}\geq \min_{x\in S_A} \frac{x^\intercal Ax}{\norm{x}^2}+ \min_{x\in S_B} \frac{x^\intercal Bx}{\norm{x}^2}= \lambda_i(A)+\lambda_j(B),$$
hence the proof is complete.
\end{proof}
Moreover, we take advantage of the following simple lemma which is also known as extremal partial trace. A proof of it can be found in \cite{Tao10}. 
\begin{lemma}
\label{lem:minquadform}
Let $L\in \Re^{d\times d}$ be a symmetric matrix.  Then for any integer $n \leq d$, 
$$ \min_{\Pi \in \mathcal{P}_n } \langle \Pi,L \rangle =\sum_{d-n+1}^d \lambda_i(L).$$ 
\end{lemma}
In particular, we use it to conclude that if $x_1,\ldots,x_n \in \Re^d$ are orthonormal vectors, then
$$ \sum_{i=1}^n x_i^\intercal Lx_i \geq \sum_{d-n+1}^d \lambda_i(L).$$ 
For a matrix $A$, we use $\sigma_1(A)\geq \ldots \geq \sigma_d(A)\geq 0$ to denote singular values of $A$ (for symmetric matrices they are the same as eigenvalues). Given a matrix, we use the following simple lemma to construct a symmetric matrix whose eigenvalues are the singular values of the input matrix and their negations.

Many of the matrices that we work with in this paper are not symmetric.
Define a symmetrization operator $\cS_d: \Re^{d \times d} \to\Re^{2d\times 2d}$ where for any matrix
 $A \in \mathbb{R}^{d\times d}$, 
$$\cS_d(A)= \left[ {\begin{array}{cc}
   0 & A \\
   A^\intercal & 0 \\
  \end{array} } \right].
$$
When the dimension is clear in the context, we may drop the subscript $d$. The following fact is immediate.
\begin{fact}
\label{fact:symmeigs}
For any matrix $A\in\Re^{d\times d}$, $\cS(A)$  has eigenvalues $\sigma_1(A)\geq \ldots\geq\sigma_d(A) \geq -\sigma_d(A) \ldots \geq -\sigma_1(A).$
\end{fact}
\begin{proof}
Let $u_1,\ldots,u_d$ and $v_1,\ldots,v_d$ be right and left singular vectors of $A$, respectively. Then we have $Au_i=\sigma_iv_i$ and $A^\intercal v_i=\sigma_i u_i$ for any $1\leq i\leq m$. Now, it is easy to see $[v_i\, \, u_i ]$ and $[-v_i \, \, u_i]$ are eigenvectors for $\cS(A)$ with eigenvalues $\sigma_i$ and $-\sigma_i$ for any $1\leq i \leq m$.
\end{proof}
Throughout the paper, we work with different norms for matrices.  The $\ell_2$ norm of matrix $A$, denoted by $\norm{A}_2$ or just $\norm{A}$ denotes $\max_{\norm{x}_2=1}\norm{Ax}_2$. Also, $\norm{A}_\infty=\max_{i,j}\abs{A_{i,j}}$ denotes the $\ell_\infty$ norm. The Frobenius norm of $A$ denoted by $\norm{A}_F$ is defined by $\norm{A}_F = \sqrt{\langle A,A \rangle}$. 
We use the following identity which relates Frobenius norm to singular values.
\begin{fact}
\label{fact:singvaldet}
For any matrix $A \in \Re^{d\times d}$,
$$ \norm{A}_F^2=\sum_{i=1}^d \sigma_i(A)^2.$$
\end{fact}
\paragraph{Determinant Maximization Problem.} We use the notion of \emph{determinant} of a subset of vectors as a measure of their 
diversity. From a geometric point of view, for a subset of vectors $V=\{v_1,\ldots,v_d\}\subset \Re^d$, $\det(\sum_{i=1}^d v_iv_i^\intercal)$ is equal to the square of the volume of the parallelepiped spanned by $V$. For $S,T \subseteq [d]$, Let $A_{S,T}$ denote the $\abs{S}\times\abs{T}$ submatrix formed by intersecting the rows and columns corresponding to $S$, $T$ respectively. 
The notation $\det_k$  is a generalization of determinant and is defined by 
$$\det_k(A)=\sum_{S \in \binom{[d]}{k}} \det{A_{S,S}}.$$
In particular, for vectors $v_1,\dots,v_k \in \Re^d$, $\det_k(\sum_{i=1}^k vv_i^\intercal)$ is equal to the square of the $k$-dimensional volume of the parallelepiped spanned by $v_1,\dots,v_k$. 
The problem of $k$-determinant maximization is defined as follows.
\begin{definition}[$k$-Determinant Maximization]
\label{def:detmaximization}
Let $V=\{v_1,\ldots,v_n\}\subset \Re^d$ be a set of vectors, and let $M \in \Re^{n\times n}$ be the Gram matrix obtained from $A$, i.e., $M_{i,j} = \langle v_i,v_j\rangle$. For an integer $k\leq d$, the goal of the $k$-determinant maximization problem is to choose a subset $S\subseteq V$ such that $\abs{S}= k$ and the determinant of $M_{S,S}$ is maximized. 
\end{definition}
Throughout the paper, we extensively use the Cauchy-Binet identity which states that for any integer $k \leq d$,  $B \in \Re^{k\times d}$, and $C \in \Re^{d \times k}$ we have 
$$ \det(BC) = \sum_{S \in \binom{d}{k}} \det(B_{[k],S})\det(C_{S,[k]}),$$
For any $S \subset [n] (|S|=k)$, if we let $V_S \in \Re^{k\times d}$ be the matrix with $\{v_i\}_{i \in S}$ as its rows, then we have
\begin{equation}
\label{eq:geometricaldetk}
\det(M_{S,S}) = \det(V_SV_S^\intercal) =\det(V_S^\intercal V_S)= \det_k (\sum_{v \in S} vv^\intercal),
\end{equation}
where the last equality uses the Cauchy-Binet identity.
The $k$-determinant maximization is also known as \emph{maximum volume} $k$-simplex since $\det_k (\sum_{v \in S} vv^\intercal)$
 is equal to the square of the volume spanned by $\{v_i\}_{i \in S}$.
Throughout the paper, we also use  the following identity which can be derived from the  Cauchy-Binet formula when the columns of $B \in \Re^{d\times n}$ are $v_i$s and $C=B^\intercal$.
\begin{equation}
\label{eq:appofcauchybinet}
\det(\sum_{i=1}^n v_iv_i^\intercal) = \sum_{S\in \binom{[n]}{d}} \det(\sum_{i\in S} v_iv_i^\intercal).
\end{equation}
We use it to deduce the following simple lemma 
\begin{lemma}
\label{lem:CauchyBinet}
For any set of vectors $v_1,\dots,v_n\in \Re^d$ and any integer $1\leq k\leq d$,
$$\det_k \left(\sum_{i=1}^n v_iv_i^\intercal\right) = \sum_{S \in {[n]\choose k}} \det_k\left(\sum_{i \in S} v_iv_i^\intercal\right)$$
\end{lemma}
\begin{proof}
For a set $T \subset [d]$ and any $1 \leq i \leq n$, let $v_{i,T}\in \Re^k$ denote the restriction of $v_i$ to its coordinates in $T$. The proof can be derived as follows
\begin{equation*}
\begin{aligned}
\det_k \left(\sum_{i=1}^n v_iv_i^\intercal\right) &= \sum_{T \in \binom{[d]}{k}} \det\left(\sum_{i=1}^n v_{i,T}v_{i,T}^\intercal \right)& \text{By definition of $\det_k$}\\
&=\sum_{T \in \binom{[d]}{k}} \sum_{S \in \binom{[n]}{k}} \det\left(\sum_{i \in S}v_{i,T}v_{i,T}^\intercal\right)& \text{By \eqref{eq:appofcauchybinet}}\\
&=\sum_{S \in \binom{[n]}{k}}\left(\sum_{T \in \binom{[d]}{k}} \det\left(\sum_{i \in S}v_{i,T}v_{i,T}^\intercal\right)\right)= \sum_{S \in \binom{[n]}{k}} \det_k\left(\sum_{i \in S}v_iv_i^\intercal \right)&
\text{By definition of $\det_k$}
\end{aligned}
\end{equation*}
\end{proof}
We also use the following identities about the determinant of matrices.
For a $d\times d$ matrix $A$, we have  $$\det(A)=\prod_{i=1}^d 
\sigma_i(A).$$ If $A$ is lower(upper) triangular, i.e. $A_{i,j}=0$ for $j>i(j<i)$, 
 we have $\det(A) = \Pi_{i=1}^d A_{i,i}$.

\subsection{Core-sets}
The notion of core-sets has been introduced in \cite{AHV-aemp-04}. Informally, a core-set for an optimization problem is a subset of the data with the property that solving the underlying problem on the core-set gives an approximate solution for the original data. This notion is somewhat generic, and many variations of core-sets exist.

The specific notion of \emph{composable core-sets} was explicitly formulated in \cite{IMMM-ccdcm-14}.

\begin{definition}[$\alpha$-Composable Core-sets]
 A function $c(V)$ that maps the input set $V \subset \Re^d$ into one of its subsets is called an $\alpha$-composable core-set for a maximization problem with respect to a function $f \colon 2^{\Re^d}\to \Re$ if, for any collection of sets $V_1,\cdots, V_m\subset \Re^d$, we have
 \[ f(c(V_1)\cup\dots\cup c(V_m))\geq \frac1{\alpha} f(V_1\cup\dots\cup V_m) \]
\end{definition}

For simplicity, we will often refer to the set $c(P)$ as the core-set for $P$ and use the term ``core-set function'' with respect to $c(\cdot)$. The {\em size} of $c(\cdot)$ is defined as the smallest number $t$ such that $c(P) \le t$ for all sets $P$ (assuming it exists). Unless otherwise stated, whenever we use the term ``core-set'', we mean a composable core-set.

\section{Spectral Spanners}
\label{sec:spectralspannerdef}
In this section we introduce the notion of \textit{spectral spanners} and review their properties. In the following, we  define the special case of spectral spanners. Later in Definition \ref{def:kspectralspanner}, we introduce its generalization, \emph{spectral $k$-spanners}. 
\begin{definition}[Spectral Spanner] 
Let $V \subset \Re^d$ be a set of vectors. We say \mbox{$U\subseteq V$} is an  $\alpha$-spectral $d$-spanner  for $V$  if for any $v \in V$, there exists a probability distribution $\mu_v$ on the vectors in $U$  so that
\begin{equation}
\label{eq:defofspannerd}
vv^\intercal \preceq   \alpha\cdot  \EE{u \sim \mu_v}{uu^\intercal}.
\end{equation}
\end{definition}
We study spectral spanners in Section \ref{sec:spectralspannerford}, and propose polynomial time algorithms for finding $\tilde{O}(d)$-spectral spanners of size $d$.
Considering \eqref{eq:defofspannerd} for all $v \in V$ implies that if $U\subseteq V$ is an $\alpha$-spectral spanner of $V$, then for any probability distribution $\mu:V\to \Re^+$, there exists a distribution $\tilde{\mu}:U \to \Re^+$ such that 
\begin{equation}
\label{eq:dspannerapp}
\EE{v \sim \mu}{vv^\intercal} \preceq \alpha \cdot \EE{u \sim \tilde{\mu}}{uu^\intercal}.
\end{equation}
We crucially take advantage of this property in Section \ref{sec:app} to 
develop core-sets for the \emph{experimental design} problem. Let $f:\mathcal{S}_d^+\to \Re^+$ be a monotone function such that $f(A) \leq f(B)$ if 
$A \preceq B$.  Roughly speaking,
we use monotonicity of $f$ along \eqref{eq:dspannerapp} to reduce optimizing  $f$ on the set of all 
matrices of the form $\EE{v \sim \mu}{vv^\intercal}$ for some distribution $\mu$,  to 
optimizing it on distributions which are only  supported on the small set $U$.
A wide range of matrix functions used in practice lie in the category of 
monotone functions, e.g. determinant, trace. 
More generally one can see $\lambda_i(.)$ for any $i$ is a monotone function, and consequently the same 
holds for any elementary symmetric polynomial of the eigenvalues.
For polynomial functions of the lower-degree, e.g. trace,  $\det_k$,  the 
monotonicity can be guaranteed by weaker constraints. Therefore, 
one should expect to find smaller core-sets with better guarantees for those 
functions. Motivated by this, we  introduce the notion of spectral $k$-spanners. Let us first define the notation $\preceq_k$ to generalize $\preceq$.
\begin{definition}[$\preceq_k$ notation]
For  two matrices $A,B \in \Re^{d\times d}$ , we say $A  \preceq_k B$ if for any $\Pi \in \mathcal{P}_{d-k+1}$, 
 we have $ \langle A,\Pi\rangle \leq \langle B,\Pi\rangle$. 
\end{definition}
\noindent In particular 
note that $A\preceq_d B$ is equivalent to $A\preceq B$  and $A \preceq_1 B$ is the same as 
$\Tr(A) \leq \Tr(B)$, since $\mathcal{P}_1=\Re^d$ and $\mathcal{P}_d=I$. More generally, the following lemma can be used to check if $A \preceq_k B$. 
\begin{lemma}
Let $A,B \in \Re^{d\times d}$ be two symmetric matrices. Then $A \preceq_k B$ if and only if $\sum_{i=k}^{d} \lambda_i(B-A) \geq 0$.
\end{lemma}
\begin{proof}
Suppose that $A \preceq_k B$. Then by definition for any $\Pi \in P_{d-k+1}$, $\langle B-A,\Pi \rangle \geq 0$, so combining with  \autoref{lem:minquadform}, we get 
$$0 \leq \min_{\Pi \in \mathcal{P}_{d-k+1}} \langle  B-A, \Pi \rangle = \sum_{i=k}^d \lambda_i(B-A).$$
The other side can also be verified in the exactly reverse order.
\end{proof}
Now, we are ready to define spectral $k$-spanners.
\begin{definition}[Spectral $k$-Spanner]
\label{def:kspectralspanner}
Let $V \subset \Re^d$ be a set of vectors. We say \mbox{$U\subseteq V$} is an  $\alpha$-spectral $k$-spanner  for $V$  if for any $v \in V$, there exists a probability distribution $\mu_v$ on the vectors in $U$  so that
\begin{equation}
\label{eq:defofspanner}
vv^\intercal \preceq_k   \alpha\cdot  \EE{u \sim \mu_v}{uu^\intercal}.
\end{equation}
\end{definition}
We may drop $k$, whenever it is clear from the context.
Finally, we  remark that  spectral $k$-spanners have the composability property: If $U_1,U_2$ are $\alpha$-spectral spanners of $V_1,V_2$ respectively, then $U_1\cup U_2$ is an $\alpha$-spectral spanner of $V_1\cup V_2$. 
This property will be useful  to construct composable core-sets.


We will prove the first part of \autoref{thm:mainspannerfork} in Section \ref{sec:spannerk}.
The second part of the theorem shows almost optimality of our results: We cannot get better than an $\Omega(d)$-spectral spanner in the worst case unless the spectral spanner has size {\em exponential} in $d$.
Next, here we prove the second part of the theorem. 

First, let us prove the claim for $k=d$.  Let $V$ be a set of $\frac12 e^{d^{\epsilon}/8}$ independently chosen 
 random  $\pm 1$  vectors in $\Re^d$. By Azuma-Hoeffding inequality and the union bound, we get that 
$$\Pr\left[\forall u,v\in V: |\langle u,v\rangle| \leq \sqrt{\frac12d^{1+\eps}}\right] \geq 1-|V|^2 e^{-d^\eps/4} \geq 1/2.$$
So, let $V$ be a set where for all $u,v\in V$, $\langle u,v\rangle^2 \leq \frac12 d^{1+\eps}$. We claim that any $d^{1-\eps}$-spectral spanner of $V$ must have all $V$. Let $U$ be such a spanner and suppose $v\in V$ is not in $u$. We observe that $vv^\intercal \not\preceq d^{1-\eps} \mE_{u\sim\mu} uu^\intercal$ for any $\mu$ supported on $U$. This is because for any $\mu$ supported on $U$,
$$ \mE_{u\sim\mu} \langle v,u\rangle^2 \leq \mE_{u\sim\mu} \frac12 d^{1+\eps} \leq \frac12 d^{1+\eps} = \frac1{2d^{1-\eps}}d^2=\frac1{2d^{1-\eps}}\langle v,v\rangle^2$$
as desired.

Now, let us extend the above proof to $k<d$. Firstly, we construct a set $V\subseteq \Re^k$ of $\frac12 e^{k^\eps/8}$ independently chosen random $\pm 1$ vectors in $\Re^k$. By above argument $V$ has no $k^{1-\eps}$-spectral $k$-spanner. Now define $V'\subseteq \Re^d$ by  appending  $d-k$ zeros to each vector in $V$. It is not hard to see that any $\alpha$-spectral $k$-spanner of $V$ is also an $\alpha$-spectral $k$-spanner of $V'$. Therefore, any $k^{1-\eps}$-spectral $k$-spanner of $V'$ has all vectors of $V'$.

\section{Spectral Spanners in Full Dimensional Case}
\label{sec:spectralspannerford}
In this section we prove \autoref{thm:mainspannerfork} for the case $k=d$. In this case we have a slightly better bound. So, indeed we will prove \autoref{thm:main-spectralspanner}.
As alluded to in the introduction we design a greedy algorithm that can be seen as a spectral analogue of the classical greedy algorithms for finding combinatorial 
spanners in graphs. The details of our algorithm is in Algorithm~\ref{alg:d-spanner}.

\begin{algorithm}[H]
\SetKwProg{Fn}{Function}{}{}
\DontPrintSemicolon
\Fn {Spectral d-Spanner($V$,$\alpha$)}{
	Let $U=\emptyset$\;
	While there is a vector $v\in V$ such that
the polytope $$ P_v=\{x| \hspace{2mm} \forall u\in U, \langle x,v\rangle > \sqrt{\alpha} |\langle x,u \rangle |\}$$
is nonempty. Find such vector $v$ and  
let $x$ be any point in $P_v$. Add $\text{argmax}_{u \in V} \langle u,x \rangle^2$ to $U$\;
  If there are no such vectors, terminate the algorithm and output $U$. 
}
\caption{Finds an $\alpha$-spectral $d$-spanner}
\label{alg:d-spanner}
\end{algorithm}\vspace{0.2cm}

Note that for any vector $v$ we can test whether $P_v$ is empty  using linear program. Therefore, the above algorithm runs in time polynomial in $|V|$ and $d$.

As alluded to in Section \ref{sec:techniques}, we first prove that our algorithm constructs a \wspa{$\tilde{O}(d)$}. 
Let us recall the definition of \wsp. 
\begin{definition}[\WSP]
A subset $U \subseteq V \subset \Re^d$ is a \wspa{$\alpha$} of $V$, if for all $v\in V$ and $x\in \Re^d$ there is a probability distribution $\mu_{v,x}$ on $U$ such that 
$$ \langle v,x\rangle^2 \leq \alpha\cdot\mE_{u\sim\mu_{v,x}}\langle u,x\rangle^2 \hspace{0.5cm} \text{equiv} \hspace{0.5cm} \langle v,x\rangle^2 \leq \alpha \cdot \max_{u\in U}\langle u,x\rangle^2$$
\end{definition}
In the rest of this section, we may call spectral spanners {\em strong} to emphasize its difference from weak spectral spanners defined above.
The rest of this section is organized as follows: In Section \ref{subsection:analdirheight} we prove that the output of the algorithm is a \wspa{$\alpha$}  of size $O(d \log d)$ for $\alpha=\Omega(d \log^2d)$. Then, in Section \ref{subsection:dirheighttospanner} we prove that for any $\alpha$, any \wspa{$\alpha$} is a strong $\alpha$-spectral spanner.

\subsection{Construction of a \WSP}\label{subsection:analdirheight}
In this section we show that Algorithm \ref{alg:d-spanner} returns an $\alpha$-spectral $d$-spanner when $\alpha$ is sufficiently larger than $d$.  At the end of this section we discuss an alternative algorithm for finding a weak spanner that achieves slightly better approximation guarantee. However, we believe that this algorithm is simpler to implement, easier to parallelize and can be tuned for practical applications. 
\begin{proposition}\label{prop:weakspanner}
There is a universal constant $C>0$ such that for $\alpha\geq C\cdot d\log^2d$, Algorithm \ref{alg:d-spanner} returns a \wspa{$\alpha$} of size $O(d\log d)$. 	
\end{proposition}

First, we observe that for {\em any} $\alpha$, the output of the algorithm is a \wspa{$\alpha$}.
For the sake of contradiction, suppose the output set $U$ is not a \wspa{$\alpha$}. So, there is a vector $v\in V$ and $x\in\Re^d$ such that
\begin{equation}
\label{eq:greedyalganal}
\langle x,v \rangle^2 > \alpha\cdot \max_{u \in U} \langle x,u \rangle^2
\end{equation}
We show that $P_v$ is non-empty, which implies $U$ cannot be the output.
We can assume $\langle x,v \rangle > 0$, perhaps by multiplying $x$ by a $-1$. So the above equation is equivalent to  
$\langle x,v \rangle > \sqrt{\alpha} \cdot \max_{u \in U}\abs{\langle u,x \rangle},$
which implies $x \in P_v$.

It remains to bound the size of the output set $U$.
As alluded to in the introduction, the main technical part of the proof is to show that the rank of lower triangular matrices is robust under small perturbations. To bound the size of $U$ we will construct such a matrix and we will use \autoref{lem:approxlowtri} (see below) to bound its rank.
Let $u_1,u_2,\dots,u_m$ be the {\em sequence} of vectors added to our spectral spanner in the algorithm, i.e., $u_i$ is the $i$-th vector added to the set $U$.  By Step 2 of the algorithm for any $u_i$ there exists a ``bad'' vector $x_i \in \Re^d$ such that 
$$ \langle u_i,x_i \rangle^2 > \alpha\cdot \max_{1 \leq j <i} \langle u_j,x_i \rangle^2,$$ 
Furthermore, by construction, $u_i$ is the vector with largest projection onto $x_i$, i.e., $u_i=\argmax_{u \in V}\langle x_i,u 
\rangle^2$. Define inner product matrix $M \in \Re^{m \times m}$ 
$$M_{ij}=\langle u_i,x_j\rangle. $$
By the above conditions on the vectors $u_i, x_j$, $M$ is diagonally dominant and 
for all $1\leq i\leq m$ and $1\leq j<i$ we have $M_{j,i} \leq \frac{M_{i,i}}{\sqrt{\alpha}}$. So the assertion  of the \autoref{lem:approxlowtri} holds for $M$ and $\eps= \frac{1}{\sqrt{\alpha}}$. By the lemma,
$$ \rank(M) \geq C \cdot \min \left\{ \frac{4\alpha}{\log^2 \alpha}, \frac{m}{\log m} \right\},$$
where $C$ for some constant $C>0$.
But, it turns out that $\rank(M)\leq d$  as it can be written as the product of an $m\times d$ matrix and a $d\times m$ matrix. 
Setting $\alpha = \frac{d\log^2d}{C}$, implies $|U|=m \leq \frac{2d\log d}{C}$ for large enough $d$, as desired. It remains to prove the following lemma.

\begin{lemma}
\label{lem:approxlowtri}
Let $M \in \Re^{m\times m}$ be a diagonally dominant and approximately lower triangular matrix in the following sense
\begin{equation}
\label{eq:Mproperties_1}
 M_{j,i} \leq \eps\cdot M_{i,i}, \forall\, 1 \leq j<i \leq m,
\end{equation}
Then, there is a universal constant $C>0$ such that we  have $\rank(M) \geq C\cdot\min\left \{ \left(\frac{1}{\eps \log \frac1{\eps}}\right)^2,\frac{m}{\log m}\right\}$.
\end{lemma}
\begin{proof}
Without loss of generality, perhaps after scaling each column of $M$ by its diagonal entry, we assume $M_{i,i}=1$ for all $i$.
 Note that rank and \eqref{eq:Mproperties_1} is invariant under scaling, so it is enough to prove the statement for such a matrix. Let $M_s$ denote the top left $s\times s$ principal submatrix of $M$ for some integer $s \leq m$ that we specify later. Note that rank is monotonically decreasing under taking principal sub-matrices, so this operations does not increase the rank and showing the assertion of the lemma  on $\rank(M_s)$ proves the lemma. Furthermore, \eqref{eq:Mproperties_1} is closed under taking principal sub-matrices. We can write  $M_s=L+E$ such that 
\begin{itemize}
\item $L \in \Re^{s\times s}$ is a lower triangular matrix where $L_{i,i}=1$ and $\abs{L_{i,j}}\leq 1$, for any $1\leq j \leq i \leq s$. In particular, $\norm{L}_{\infty}=1$.
\item $\norm{E}_{\infty}\leq \eps$ (note that we may further assume $E$ is upper triangular, but we do not use it in our proof).
\end{itemize}
Let $\sigma_1(M_s) \geq \ldots \geq \sigma_s(M_s)$ denote singular values of $M_s$. Obviously, $\sigma_{i}(M_s)>0$ implies $\rank(M_s) \geq i$ for any $1 \leq i \leq s$.
Considering this fact, let us give some intuition on why $M_s$ has a  large rank. Since $L$ is lower-triangular with non-zero entries on the diagonal, it  is a full rank matrix. Moreover, entries of $E$ are much smaller than (diagonal) entries of $L$. Singular values of $E$ are on average much smaller than those of $L$, so adding $E$ to $L$ can only make a small fraction of singular values of  $L$ vanish. This implies that $M_s=L+E$ must have a high rank.
Now we make the argument rigorous.  

Let $\cS(M_s),\cS(L),\cS(E)$ be the symmetrized versions of  $M_s,L$ and $E$ respectively (see \autoref{subsection:linalg}). 
By \autoref{fact:symmeigs}, to show $\sigma_{i}(M_s)>0$  for some $i$, we can equivalently prove $\lambda_{i}(\cS(M_s)) >0$. We use
Lemma \ref{lem:boundingeigvals}: Setting $A=\cS(L)$ and $B=\cS(E)$,   for any pair of integers $\ell<k\leq s$ such that 
\begin{equation}
\label{eq:findkl}
\lambda_{k}(\cS(L)) +\lambda_{2s-\ell}(\cS(E)) > 0
\end{equation}
we have $\lambda_{k-\ell}(\cS(M_s))>0$.
So to prove the lemma, it suffices to find $s, k$ and $\ell$ satisfying the above and $k-\ell \geq C\cdot\min\left \{ \left(\frac{1}{\eps \log \frac1{\eps}}\right)^2,\frac{m}{\log m}\right\}$ for some constant $C$.

To find proper values of $k$ and $\ell$, we use the following two claims. 
\begin{claim}
\label{claim:1}
For any $\ell \leq s$, 
$$ \lambda_{2s-\ell}(\cS(E)) \geq \lambda_{2s-\ell+1}(\cS(E))=-\sigma_\ell(E) \geq \frac{-\norm{E}_F}{\sqrt{\ell}}\geq \frac{-\eps\cdot s}{\sqrt{\ell}}.$$
\end{claim}

\begin{claim}
\label{claim:2}
For any $k<\frac{s}{2}$, 
$$ \lambda_k(\cS(L)) = \sigma_k(L) \geq \left(\frac{k-1}{s^2}\right)^{\frac{k-1}{s}}.$$
\end{claim}
Therefore, to show \eqref{eq:findkl} it is enough to show
\begin{equation}
\label{eq:findkl2}
s \log \frac{\sqrt{\ell}}{\eps s} > (k-1)\log \frac{s^2}{k-1},
\end{equation}
for $k,\ell \leq \frac{s}{2}$.
We analyze the above in two cases. If $m\log m \leq \frac{1}{\eps^2}$, then one 
can see that for $s=m$, $k=\lfloor\frac{m}{4\log m}\rfloor$ and 
$\ell=\lfloor\frac{m}{8\log m}\rfloor$, and large enough $m$, \eqref{eq:findkl2} 
holds. It implies that in this case $\rank(M_s) \geq k-\ell \geq \frac{m}{8\log 
m}$, thus we are done. Now suppose that $\frac1{\eps^2} \leq m\log m$. We set 
$s<m$ to be the largest integer such that $s\log s\leq \frac{1}{16\eps^2}$. Next, 
we let $\ell=\lfloor 4\eps^2s^2\rfloor$. Note that $s\log s\leq \frac{1}{16\eps^2}$ 
implies $\ell \leq \frac{s}{4\log s}$. Now applying $\ell=\lfloor 4\eps^2s^2\rfloor$ 
into \eqref{eq:findkl2} turns it into 
\begin{equation*}
\begin{aligned}
s > (k-1)\log \frac{s^2}{k-1}.
\end{aligned}
\end{equation*}
So, for $k= \lfloor \frac{s}{2\log s}\rfloor$, the above is satisfied. Furthermore, in this case $k-\ell= \lfloor \frac{s}{2\log 
s}\rfloor-\lfloor4\eps^2s^2\rfloor \geq \frac{s}{4\log s} \geq \frac{1}{256\eps^2 \log^2\frac{1}{\eps}}$,  as  $s$ is the largest number such that $s\log s\leq \frac{1}{16\eps^2}$ and $\log s \leq 2\log \frac{1}{\epsilon}$. So the lemma holds for $C \geq \frac1{256}$. 
\end{proof}

\begin{proofof}{Claim \ref{claim:1}}
By Fact \ref{fact:singvaldet} we know that 
$$\sum_{i=1}^s \sigma_i(E)^2 =\| E \|_F^2 \leq \norm{E}^2_\infty \cdot s^2 \leq \eps^2\cdot s^2.$$
Now, by Markov inequality we get 
$\sigma_{\ell}(E)^2 \leq \frac{\eps^2 s^2}{\ell}$. Therefore, the claim is proved.
\end{proofof}

\begin{proofof}{Claim \ref{claim:2}}
Since $L$ is lower-triangular, we have that
\begin{equation}
\label{eq:det}
\prod_{i=1}^s \sigma_i(L) = \det{L} = \prod_{i=1}^s L_{i,i} =1,
\end{equation}
It follows that for any $k\leq s$,
\begin{equation}
\label{eq:useofdet}
\prod_{i=1}^{k-1}\sigma_i(L) = \frac{1}{\prod_{j=k}^{s}\sigma_{j}(L)} \geq \frac{1}{\sigma_{k}(L)^{s-k+1}}.
\end{equation}
Now, we use the Frobenius norm to prove an upper bound on the first $k-1$ singular values. 
By Fact \ref{fact:singvaldet},
\begin{equation} 
\label{eq:frobeniusnorm}
\sum_{i=1}^{s}\sigma_i(L)^2 = \|L\|_F^2 \leq \norm{L}_{\infty}\cdot s^2= s^2,
\end{equation}
By AM-GM inequality we get
$$ \prod_{i=1}^{k-1}\sigma_i(L) \leq \left(\frac{\sum_{i=1}^{k-1}\sigma_i(L)^2}{k-1}\right)^{\frac{k-1}{2}} \leq \left(\frac{\norm{L}_F^2}{k-1}\right)^{\frac{k-1}{2}} \leq \left (\frac{s^2}{k-1}\right)^{\frac{k-1}{2}}.$$
The above together with \eqref{eq:useofdet} proves $\sigma_k(L) \geq \left(\frac{k-1}{s^2}\right)^\frac{k-1}{2(s-k+1)}$. Noting that for $k \leq \frac{s}{2}$, $2(s-k+1)\geq s$ completes the proof  of the claim.
\end{proofof}

\medskip
We would like to thank an anonymous reviewer for suggesting an alternative algorithm for finding weak spanners. It offers slightly better guarantees, and finds a \wspa{$d$} 
 of size $O(d\log d)$.  However, as we argue, our algorithm can be made more efficient in practice, and in particular in a distributed setting.
\paragraph{An alternative 
Algorithm.}   For a 
subset $V \in \mathbb{R}^d$, define 
$\text{sym}(V)$ to be the symmetric 
set $\text{sym}(V)=V \cup \{-x | x 
\in V\}$. A direct application of  the separating hyperplane theorem shows that a subset $U \subseteq V$ is a \wspa{$\alpha$} of $V$, if 
$\text{conv}(\text{sym}(V))\subseteq \sqrt{\alpha}\cdot
\text{conv}(\text{sym}(U))$ where $\text{conv}$ refers to the convex hull of the set. Knowing this, we can apply the celebrated result of
F. John \cite{ball1997elementary} to get a \wspa{d}. Letting the notation MVEE of a set denote the minimum volume ellipsoid enclosing the set, it implies that there exists a subset $U \subset V$ of size $O(d^2)$ and an ellipsoid $E$  where $E=\text{MVEE}(\text{sym}(U))=\text{MVEE}(\text{sym}(V))$ and $\frac{E}{\sqrt{d}} \subseteq \text{conv}(\text{sym}(U))$. Therefore, $U$ is a \wspa{d} of $V$ with size $O(d^2)$. 
Moreover, the size of $U$ can be reduced to $O(d)$ by using the spectral sparsification machinery of \cite{batson2012twice}. We also note that these ideas has been extended in \cite{barvinok2014thrifty} to show that for any $d$-dimensional symmetric convex body $C$ and any $0 < \epsilon <1$, there is a polytope $P$ with roughly $d^{\frac{1}{\epsilon}}$ vertices such that $P \subset C \subset (\sqrt{\epsilon d})P$. In our terminology, it gives an algorithm to find a weak $d$-spectral spanner of $V$ of size $O(d)$. Although, the approximation guarantee can be improved by a 
log factor in this algorithm, this improvement comes at a cost.
First of all, finding the John's 
ellipsoid requires solving a semidefinite program  with $O(d^2)$ variables whereas 
 in Algorithm \ref{alg:d-spanner}, we only need to solve linear programs with $O(d)$ many variables. This requires polynomially smaller amount of memory.
Furthermore, note that the main computational 
task of each step of Algorithm \ref{alg:d-spanner} is to solve $|V|$ 
feasibility LPs where each of 
them has  $\tilde{O}(d)$ variables and constraints. These LPs can be solved in parallel:  having access to $O(|V|)$ many processors, our greedy algorithm runs in $\text{poly}(d)$ time in PRAM model of computation. This extreme parallelism cannot be achieved using the above approach.
Finally, in finding the weak spanner one can {\em tune} the value of $\alpha$ in \eqref{eq:tunealpha} based on the structure of the given data points and the ideal size of the core-set, making the algorithm more suitable for applications.

\subsection{From Weak Spectral Spanners to Strong Spectral Spanners}
\label{subsection:dirheighttospanner}
In this section, we prove that  if $U$ is a 
\wspa{$\alpha$} of $V$, then it is a strong 
$\alpha$-spectral spanner of $V$. Combining 
with \autoref{prop:weakspanner} it proves 
\autoref{thm:main-spectralspanner}. 
\begin{lemma}
\label{lem:CoresetToSpannerD}
For any set of vectors $V \subset \Re^d$, any \wspa{$\alpha$} of $V$ is a strong $\alpha$-spectral spanner of $V$. 
\end{lemma}
\begin{proof}	
Let $U$ be a \wspa{$\alpha$} of $V$.
Fix a vector $v\in V$, we write a program to find a probability distribution $\mu_v:U\to \Re^+$ such that 
    $vv^T  \preceq \frac1{\delta} \cdot \EE{u \sim \mu_v}{uu^T}$,
for the largest possible $\delta$.
It turns out that this is a semi-definite program, where we have a variable $p_u=\Pr_{\mu_v}(u)$ to denote the probability of each vector $u\in U$, see \eqref{CP:spanner} for details.
\begin{equation}
				\begin{aligned}
				\max \hspace{.7cm} & \delta \\
				\text{s.t } \hspace{.7cm} &  \delta\cdot vv^T \preceq \EE{u \sim \mu_v}{uu^T} \\
                \hspace{.7cm} & \mu_v \text{ is a distribution on } U 
                \end{aligned}
                \label{CP:spanner}
\end{equation}

To prove the lemma, it suffices to show the optimal of the  program is at least $\frac{1}{\alpha}$.
To do that, we analyze the dual of the program.
We first show the set of feasible solutions of the program has 
	a non-empty interior; this implies that the Slater condition is satisfied, and the duality gap is zero. Then we show any solution of the dual has value at least $1/\alpha$.
	
	To see 
	the first assertion, we let $\mu_v$ be equal to the uniform distribution on $U$ and $\delta \leq \frac{1}{\alpha|U|}$. It is not hard to see that this is a 
	feasible solution of the program since  $U$ is a \wspa{$\alpha$}.
	
	Next, we prove the second statement. First we write down the dual.	\begin{align*}
	\min \hspace{.5cm} &\lambda \hspace{1.5cm}\\
	\text{s.t.}\hspace{.5cm} &u^{T}Xu\leq \lambda, \, \forall u \in U\\
    \hspace{.5cm} &v^TXv \geq 1\\
    \hspace{.5cm} &X \succeq 0
	\end{align*}
 Let $(X,\lambda)$ be a feasible solution of the dual. Our goal is to show $\lambda \geq \frac{1}{\alpha}$.
Let $E=\{x \in 
	\mathbb{R}^d \, | \, x^TXx \leq \lambda \}$ be an   ellipsoid of radius $\sqrt{\lambda}$ defined by 
	$X$. The set $E$ has the following properties:
	\begin{itemize}
	\item Convexity,
	\item Symmetry: If $x\in E$, then $-x\in E$,
	\item $U\subseteq E$: By the dual constraints $u^{\intercal}Xu \leq \lambda$ for all $u \in U$.
	\end{itemize}
	Let $\bar{v}=v/\sqrt{\alpha}$. We claim that $\bar{v}\in E$. Note that if $\bar{v}\in E$ we obtain 
	$$\lambda\geq \bar{v}^\intercal X \bar{v} \geq \frac{1}{\alpha},$$ 
	which completes the proof. 
	
	For the sake of contradiction suppose $\bar{v} \notin E$. We show that $U$ is not a \wspa{$\alpha$}. 
	By convexity of $E$  there is a hyperplane  separating $\bar{v}$ from $E$. So there is a vector $e \in \Re^d$ such that
	\begin{align*}
 \langle v,e\rangle = \sqrt{\alpha}\cdot \langle \bar{v},e \rangle  \geq \sqrt{\alpha} \hspace{1cm} \text{and} \hspace{1cm}
	\forall x\in E, \hspace{1mm}\langle x,e \rangle <1.
	\end{align*}
	Moreover, by symmetry of $E$, for any $x\in E$,
	$$\langle x,e \rangle ^2 \leq \max\{\langle x,e\rangle, \langle -x,e\rangle\}^2 <1$$ 
Finally, since $U \subset E$, we obtain 
$\max_{u\in U} \langle u,e \rangle^2 <1.$
Therefore, $\langle v,e\rangle ^2\not\leq \alpha\max_{u\in U}\langle u,e\rangle^2$ which implies $U$ is not a \wspa{$\alpha$}.
\end{proof}


\section{Construction of Spectral $k$-Spanners}
\label{sec:spannerk}
In this section we extend our proof on spectral $d$-spanner to spectral $k$-spanners for $k<d$, this proves our main theorem \ref{thm:mainspannerfork}. 
Here is our high-level plan of proof:  First we use
the greedy algorithm of \cite{cm-smvsm-09} for volume maximization to 
 find an $\tilde{O}(k)$-dimensional linear subspace 
of $\Re^d$ onto which input vectors  have a ``large'' projection. Next, we 
apply  \autoref{thm:main-spectralspanner} to this $\tilde{O}(k)$-dimensional space to obtain the desired spectral $k$-spanner.
\subsection{Greedy Algorithm for Volume Maximization}
In this subsection, we prove the following statement.
\begin{proposition}
\label{lem:orthogonalpartofv}
For any set of vectors $V \subset \Re^d$, and any $k<d$ and $m>2k$, there is a set $U\subseteq V$ of size $m$ such that for all $v\in V$ we have 
\begin{equation}\label{eq:vUperpkspanner} v_{U^\perp} v_{U^\perp}^\intercal \preceq_k 2m^{(\frac{2k}{m})} \cdot \mE_{u\sim\mu} uu^\intercal,
\end{equation}
where $v_{U^\perp}=\Pi_{\langle U\rangle^\perp}(v)$ is the projection of $v$ on the space orthogonal to the span of $U$ and $\mu$ is the uniform distribution on the set $U$.
\end{proposition}
As we will see in the next subsection, for $m=\Theta(k\log k)$, the set $U$ promised above will be a part of our spectral $k$-spanner. 
Roughly speaking, to obtain a spectral spanner of $V$,   it is enough to additionally add a spectral spanner of $\{\Pi_{\langle U\rangle}(v)\}_{v\in V}$. In the next subsection, will use \autoref{thm:main-spectralspanner} for the latter part.

First, we will describe an algorithm to find the set $U$ promised in the proposition. Then, we will prove the correctness.
We use the greedy algorithm of \cite{ccivril2009selecting} for volume maximization to find the set $U$.
\begin{algorithm}[H]
\SetKwProg{Fn}{Function}{}{}
\DontPrintSemicolon
\Fn{Volume Maximization($V$,$m$)}{
Let $U=\emptyset$\;
While $|U| <m$, add $\argmax_{v\in V} \norm{\Pi_{\langle U\rangle^\perp}(v)}$ to $U$\;
Return $U$.
}
\caption{Greedy Algorithm for Volume Maximization}
\label{alg:greedy}
\end{algorithm}
Let $U=\{u_1,\dots,u_m\}$ be the output of the algorithm and suppose $u_i$ is the $i$-th vector added to the set, and $\mu$ be a uniform distribution on $U$. Fix a vector $v\in V$ for which we will verify the assertion of the proposition. Note that if $v\in U$ the statement obviously holds. So, assume $v\notin U$.

Fix a $\Pi\in\Pi_{d-k+1}$. Observe that $\langle v_{U^\perp}v_{U^\perp}^\intercal,\Pi\rangle \leq \norm{\Pi_{\langle U\rangle^\perp}(v)}^2$. On the other hand, 
by \autoref{lem:minquadform}, 
$\left\langle \mE_{u\sim\mu}  uu^\intercal,\Pi\right\rangle \geq \sum_{i=k}^d \lambda_i$ where $\lambda_1 \geq \lambda_2 \geq \ldots \geq \lambda_d$ are eigenvalues of $\mE_{u\sim\mu} uu^\intercal$.
Therefore, to prove \eqref{eq:vUperpkspanner}, it suffices to prove  
\begin{equation}
\label{eq:greedyvolgoal}
\norm{\Pi_{\langle U\rangle^\perp}(v)}^2 \leq 2m^{\frac{2k}{m}}\cdot \sum_{i=k}^d \lambda_i.
\end{equation}

Define $\hu_1,\hu_2,\ldots,\hu_m$ to be an orthonormal basis of $\langle U\rangle$ obtained by the Gram-Schmidt process on $u_1,\dots,u_m$, i.e., $\hu_1=\frac{u_1}{\norm{u_1}}$, $\hu_2=\frac{\Pi_{\langle u_1\rangle^\perp}(u_2)}{\norm{\Pi_{\langle u_1\rangle^\perp}(u_2)}}$ and so on. Define $M\in \Re^m$ to be a matrix where 
the $i_{\text{th}}$ column is the representation of $u_i$ in the orthonormal basis formed by $\{\hu_1,\ldots,\hu_m\}$, i.e., for all $1\leq i,j\leq m$,
$$M_{i,j}= \langle u_j,\hu_i\rangle.$$   
Note that $\mE_{u\sim\mu} uu^\intercal$ is the same as $\frac1m MM^\intercal$ up to a rotation of the space.
In other words, both matrices have the same set of non-zero eigenvalues.
Since eigenvalues of $\frac1m M M^\intercal$ are the squares of the singular values of $\frac1{\sqrt{m}}M$,  to prove \eqref{eq:greedyvolgoal} it is enough to show
\begin{equation}
\label{eq:maingoaloforthogonalpart}
\norm{\Pi_{\langle U\rangle^\perp}(v)}^2 \leq 2m^{\frac{2k}{m}}\cdot \sum_{i=k}^m \sigma_i^2(m^{-1/2} M).
\end{equation}

Since $v\in V$ and $v\notin U$ we get
$\norm{\Pi_{\langle U\rangle ^\perp}(v)}^2 \leq \norm{\Pi_{\langle \hu_1,\dots,\hu_{m-1}\rangle^\perp}(u_m)}^2 =M_{m,m}^2.$
So to prove \eqref{eq:maingoaloforthogonalpart}, it suffices to show 
\begin{equation}
\label{eq:maingoaloforthogonalpart2}
 M_{m,m}^2 \leq 2m^{\frac{2k}{m}}\cdot \sum_{i=k}^m \sigma_i^2(m^{-1/2}M)
\end{equation}

Note that the above inequality can be seen just as a property of the matrix $M$. First, let us discuss properties of $M$ that we will use to prove the above:
\begin{enumerate}[label=\Roman*)]
\item $M$ is upper-triangular as $u_i \in \langle \hu_1,\dots,\hu_i \rangle$.
\item By description of the algorithm, for any $i < j \leq m$ we have 
\begin{equation}
\label{eq:Mproperties}
M_{i,i}^2=\norm{\Pi_{\langle \hu_1,\dots,\hu_{i-1}\rangle^\perp}(u_i)}^2 \leq \norm{\Pi_{\langle \hu_1,\ldots,\hu_{i-1}\rangle^\perp}(u_j)}^2 =\sum_{\ell=i}^j M_{\ell,j}^2
\end{equation}
\end{enumerate}
The following lemma completes the proof of \autoref{lem:orthogonalpartofv}.
\begin{lemma}
Let $M\in \Re^{m\times m}$ satisfying (i) and (ii). For any $k<m/2$, we have
$$ M_{m,m}^2 \leq 2m^{\frac{2k}{m}}\sum_{i=k}^m\frac1m \sigma_i^2(M).$$	
\end{lemma}
\begin{proof}
Here is the main idea of the proof. First, we   use Cauchy-Interlacing theorem  along with property (ii)  to deduce $\sigma_i$ cannot be much larger than $M_{i,i}$. Then, we combine it with the fact that $M$ is upper triangular and so  $\det(M)= \prod_{i=1}^m M_{i,i} = \prod_{i=1}^m \sigma_i,$
to upper-bound $M_{m,m}^2$ by a multiple of $\sum_{i=k}^m \sigma_i^2(M)$.

First, we show for all $1 \leq i \leq m$, 
\begin{equation}\label{eq:sigmaiMupper} \sigma_i^2(M) \leq (m-i+1)M_{i,i}^2.	
\end{equation}
Define $M_i$ to be the $(m-i+1)\times(m-i+1)$ matrix obtained by removing the first $i-1$ rows and columns of $M$. 
First, Cauchy interlacing theorem tells us $ \sigma_i(M) \leq \sigma_1(M_i)$. Secondly, by Fact \ref{fact:singvaldet} and property (ii) we have
$$\sigma_1(M_i)^2 \leq \sum_{j=1}^{m-i+1} \sigma_j(M_i)^2 = \norm{M_i}_F^2  \leq  (m-i+1)M_{i,i}^2.$$
This proves \eqref{eq:sigmaiMupper}.
 Since $M$ is upper-triangular,
\begin{equation*}	
\det(M)^2=\prod_{i=1}^m M_{i,i}^2 = \prod_{i=1}^m \sigma_i^2(M) \leq \left(\frac{\sum_{j=k}^m \sigma_j(M)^2}{m-k+1}\right)^{m-k+1} \prod_{i=1}^{k-1} \sigma_i^2 =:\beta \prod_{i=1}^{k-1}\sigma_i^2 
\end{equation*}
where the inequality follows by the AM-GM inequality and $\beta=(\frac{\sum_{j=k}^m \sigma_j(M)^2}{m-k+1})^{m-k+1}$. 
By \eqref{eq:sigmaiMupper}, 
\begin{equation}
\label{eq:applydet}
\prod_{i=1}^m \sigma_i^2(M) \leq \beta \prod_{i=1}^{k-1} (m-i+1)M_{i,i}^2 \leq m^k\beta \prod_{i=1}^{k-1} M_{i,i}^2.
\end{equation}
Using $\prod_{i=1}^m \sigma_i^2=\prod_{i=1}^m M_{i,i}^2$ again, we get
$$ \prod_{i=k}^m M_{i,i}^2 \leq m^k\beta.$$
Using property (ii) again, we have $M_{i,i}^2 \geq M_{m,m}^2$ for all $i$. Therefore, 
$\prod_{i=k}^m M_{i,i}^2 \geq M_{m,m}^{2(m-k+1)}$, we get
$$ M_{m,m}^{2(m-k+1)} \leq m^k\beta$$
The lemma follows by raising both sides to $1/(m-k+1)$ and using that $m-k+1\geq m/2$ since $k<m/2$.
\end{proof}

\subsection{Main algorithm}
In this section we prove Theorem \ref{thm:mainspannerfork}. The details of our algorithm are described in Algorithm \ref{alg:k-spanner}.

\begin{algorithm}
\SetKwProg{Fn}{Function}{}{}
\DontPrintSemicolon
\Fn{Spectral-$k$-Spanner($V$,$\alpha$,$k$)}{
Set $m$, such that $m^{\frac{2k}{m}}=O(1)$.\;
Run Volume-Maximization($V$,$m$) of Algorithm \ref{alg:greedy} and let $U$ be the output, i.e., the set of vectors satisfying \eqref{eq:vUperpkspanner}. \;
3. Run Spectral-$d$-Spanner($\{\Pi_{\langle U\rangle}(v)\}_{v\in V}$,$\alpha$) of  Algorithm \ref{alg:d-spanner}
and let $W$ be the output of the corresponding spectral $m$-spanner. \;
4. Return $U\cup \{v: \Pi_{\langle U\rangle}(v)\in W\}$.
}
\caption{Finds an $\alpha$-Spectral $k$-Spanner}
\label{alg:k-spanner}
\end{algorithm}


 First of all let us analyze the size of the output. By definition, Algorithm 
\ref{alg:greedy} returns $m$ vectors. Then, by \autoref{thm:main-spectralspanner}, Algorithm \ref{alg:d-spanner} has size at most  
$O(m\log m)$. Since $m=O(k\log k)$, the size of the output is $|U|+|W| \leq O(k\log^2k)$, as desired. 

In the rest of this section we prove the correctness. Fix a vector $v \in V$, we need to find a distribution $\mu_v$ on $U\cup W$ such that $vv^\intercal \preceq_k \alpha \mE_{u\sim\mu_v} uu^\intercal$ for some $\alpha=\tilde O(k)$.

First, by \autoref{fact:PiCauchy},
$$ vv^\intercal \preceq_k 2(\Pi_{\langle U\rangle^\perp}(v) \Pi_{\langle U\rangle^\perp}(v)^\intercal + \Pi_{\langle U\rangle}(v)\Pi_{\langle U\rangle}(v)^\intercal)$$
 So, it is enough to prove that
\begin{equation}
\label{eq:spannergoal2}
\Pi_{\langle U\rangle^\perp}(v)\Pi_{\langle U\rangle^\perp}(v)^\intercal + \Pi_{\langle U\rangle}(v) \Pi_{\langle U\rangle}(v)^\intercal \preceq_k (\alpha/2)\mE_{u\sim\mu_v} uu^\intercal
\end{equation}
 We proceed by upper-bounding the LHS term by term. 
 By \autoref{lem:orthogonalpartofv}, 
\begin{equation}\label{eq:Uperppreceqk}\Pi_{\langle U\rangle^\perp}(v)\Pi_{\langle U\rangle^\perp}(v)^\intercal \preceq_k O(1) \cdot \mE_{u\sim \mu} uu^\intercal
\end{equation} 
 where $\mu$ is the uniform distribution on $U$.  So, to prove the above, it is enough to find a distribution $\mu_v$ on $U\cup W$ such that
\begin{equation}
\label{eq:notorthogonalpart}
\Pi_{\langle U\rangle}(v) \Pi_{\langle U\rangle}(v)^\intercal \preceq_k \alpha\mE_{u\sim\mu_v} uu^\intercal
\end{equation}
for some $\alpha=\tilde {O}(k)$. From now on, for any vector $v\in V$ we use $\hv$ to denote $\Pi_{\langle U\rangle}(v)$.

By description of the algorithm, $\{\hv\}_{v\in W}$  is an  $O(m\log^2 m)$-spectral $m$-spanner for $\{\hv\}_{v\in V}$. So, there exists a probability distribution $\nu_v$ on $W$ such that 
\begin{equation} 
\label{eq:indonsmallspace}
 \hv\hv^\intercal \preceq O(m\log^2m)) \mE_{w\sim \nu_v}\hw\hw^\intercal \hspace{0.5cm} \text{equiv} \hspace{0.5cm} \forall x\in\langle U\rangle: \langle x,\hv \rangle^2 \leq O(m\log^2 m)) \cdot \mE_{w \sim \nu_v}\langle \hw,x \rangle^2 
\end{equation}
In fact the above holds for any $x \in \Re^d$, as $\langle x,\hu \rangle =\langle \Pi_{\langle U \rangle}(x),\hu \rangle$ for any vector $u\in V$. Therefore,  for any $\Pi \in \Pi_{d-k+1}$, by summing \eqref{eq:indonsmallspace} up over an orthonormal basis of $\Pi$ and noting $m=O(k\log k)$, we get 
$$ \langle \Pi,\hv\hv^\intercal \rangle \leq  \tilde O(k) \cdot \langle \mE_{ w \sim \nu_v} \hw\hw^\intercal,\Pi \rangle,$$ 
which by definition implies 
\begin{equation}
\label{eq:indforinS}
\hv\hv^\intercal \preceq_k \tilde O(k)  \cdot \mE_{w \sim \nu_v}\hw\hw^\intercal.
\end{equation}
Therefore, to show \eqref{eq:notorthogonalpart} for $\alpha= \tilde O(k)$  it suffices to find a distribution $\mu_v$ on $U\cup W$ such that
$$ \mE_{w \sim \nu_v}\hw\hw^\intercal \preceq_k O(1)\cdot \mE_{u \sim \mu_v}uu^\intercal.$$ 
But, observe that  for any $w  \in W$, we can write
\begin{equation*}
\hw\hw^\intercal \preceq_k 2\left( ww^\intercal+\Pi_{\langle U \rangle ^\perp}(w)\Pi_{\langle U \rangle ^\perp}(w)^\intercal \right) \preceq_k O(1)\cdot(ww^\intercal + \mE_{u\sim\mu} uu^\intercal)
\end{equation*}
where $\mu$ is the uniform distribution over $U$.
  The first inequality follows by \autoref{fact:PiCauchy}
and the second inequality follows by  equation \eqref{eq:Uperppreceqk} which holds for all vectors $v\in V$.
Averaging out the above inequality with respect to the distribution $\nu_v$ completes the proof.


\section{Applications}
\label{sec:app}
In this section we discuss applications of \autoref{thm:mainspannerfork} in designing composable core-sets. 
As we discussed in the intro, we show that for many problems spectral spanners provide almost the best possible composable core-set in the worst case. Next, we see that for any function $f$ that is ``monotone'' on PSD matrices, spectral spanners provide a composable core-set for a  {\em fractional} budgeted minimization problem with respect to $f$. Later, in Sections \ref{sec:detmaximization} and \ref{sec:experimentaldesign} we see that for a large class of monotone functions the optimum of the fractional budgeted minimization problem is within a small factor of the optimum of the corresponding integral problem. So,  spectral spanners provide almost optimal composable core-sets for several spectral budgeted minimization problems.

Let $V \subset \Re^d$ be a set of  vectors. For a function  $f:\mathbb{S}_d^+\to \Re^+$ on PSD matrices and a positive integer $B$ denoting the budget, the {\em fractional budgeted minimization} problem is to
choose a mass $B$ of the vectors of $V$, i.e., $\{s_v\}_{v\in V}$ where $\sum_v s_v\leq B$, such that $f(\sum_v s_v vv^T)$ is minimized. This can be modeled as a continuous optimization problem, see \ref{CP:Application} for details.
\begin{figure*}[h]
\centering
\fbox{\parbox{4in}{ \vspace*{0mm}
				\vspace{-.17mm}
                \begin{equation*}
				\begin{aligned}
				\inf \hspace{.7cm} &  f\left(\sum_{v \in V} s_vvv^T\right).\\
				\text{s.t } \hspace{.7cm} &  \sum_{v \in V} s_v \leq B \\
                \hspace{.7cm} & s_v \geq 0,\,\, \forall v \in V 
				\end{aligned}
                \end{equation*}
				}}
\begin{center}{\bf \setword{BM}{CP:Application}}\end{center}
\end{figure*}

\begin{definition}[{$k$-monotone functions}]
We say a function $f:\mathbb{S}_d^+\to \Re^+$ is {\em $k$-monotone} for an integer  $1\leq k\leq d$, if for all PSD matrices $A,B\succeq 0$, we have $A \preceq_k B$ implies $f(A) \geq f(B)$.

	We say $f$ is {\em vector $k$-monotone} if for all PSD matrices $A,B$ and all vectors $v\in\Re^d$, if $vv^\intercal \preceq_k B$, then $f(A+vv^\intercal) \geq f(A+B)$. Note that any $k$-monotone function is obviously vector $k$-monotone as $A+vv^\intercal \preceq_k A+B$. 
\end{definition}

We show that an algorithm for finding $\alpha$-spectral $k$-spanners give an $\alpha$-composable core-set  function for the fractional budgeted minimization {\em for any} function $f$ that is vector $k$-monotone. We emphasize that our composable core-set {\em does not} depend on the choice of $f$ as long as it is vector monotone.
\begin{proposition}
\label{thm:mainoptdesign}
For any $1 \leq k \leq d$ and any vector $k$-monotone function $f: \mathbb{S}^+_d\to\Re^+$, Algorithm \ref{alg:k-spanner} gives  a $\beta(f,\tilde O(k))$-composable core-set of size $\tilde O(k)$  for the
fractional budgeted minimization problem, \ref{CP:Application}($V$,$f$,$B$),  where for any $t>0$,
$$\beta(f,t) = \sup_{A \in \mathbb{S}_d^+ }\frac{f(A)}{f(tA)}.$$ 
\end{proposition}
\begin{proof}
Let  $V_1,V_2,\dots,V_p$ be $p$  given input sets for an arbitrary integer $p$, and let $\bigcup_{i=1}^p V_i=V$. For each $1\leq i\leq p$, let  
$U_i$ be the output of Spectral $k$-Spanner($V_i$,$k$,$\alpha$). By \autoref{thm:mainspannerfork}, for  $\alpha=\tilde O(k)$, $|U_i|\leq \tilde O(k)$. Let $U=U_1\cup\dots\cup U_p$.

Fix a $k$-monotone function $f$ and a budget $B>0$ and let $\mathbf{s}=\{s_v\}_{v\in V}$ be a feasible solution of \ref{CP:Application}($V$, $f$, $B$). To prove the assertion we need to show that 
there exists a 
feasible solution $\tilde{\mathbf{s}}$ of \ref{CP:Application}($U,f,B)$ such that 
\begin{equation}
f\left(\sum_{u\in U} \tilde{s}_u uu^T\right) \leq \beta(f,\alpha)  \cdot f\left(\sum_{v \in V} s_v vv^T\right) 
\label{eq:goalapplication}
\end{equation}
By composability property of spanners, $U$ is an $\alpha$-spectral $k$-spanner of $V$. Therefore, for any $v \in V$, there exists a probability distribution $\mu_v$ on $U$ such that $vv^T \preceq_k \alpha \cdot \mE_{u\sim\mu_v}  uu^\intercal$. Say $V=\{v_1,\dots,v_m\}$. It follows by vector $k$-monotonicity of $f$ and by $U$ bein an $\alpha$-spectral $k$-spanner that
\begin{eqnarray*}{f\left(\sum_{i=1}^m s_{v_i} v_iv_i^\intercal\right)} &\geq & f\left(\alpha s_{v_1}\mE_{u\sim\mu_{v_1}}uu^\intercal +\sum_{i=2}^m s_{v_i}v_iv_i^\intercal\right) \\
	& \geq & f\left(\sum_{i=1}^2 \alpha s_{v_i}\mE_{u\sim\mu_{v_i}} uu^\intercal + \sum_{i=3}^m s_{v_i} v_iv_i^\intercal\right)\geq \dots\geq 
{f\left(\sum_{i=1}^m\alpha s_{v_i}\mE_{u\sim\mu_{v_i}} uu^T\right)}  
\end{eqnarray*}
Now, define $\tilde{\mathbf{s}}$ by $\tilde{s}_u= \sum_{v\in V} s_v\Pr_{\mu_v}[u]$ for any $u \in U$.
It is straight-forward to see that this is a feasible solution of \ref{CP:Application}($V,f,B$) since $\sum_{u\in U} \tilde{s}_u = \sum_{v\in V} \sum_{u\in U} \Pr_{\mu_v}[u] = \sum_{v\in V} s_v\leq B$. Therefore,
\begin{align*}
{f\left(\sum_{v \in V} s_v vv^\intercal\right)} \geq 
{f\left(\sum_{v \in V}\alpha s_v\mE_{u\sim\mu_v} uu^\intercal\right)} =f\left(\sum_{u\in U}\alpha\tilde{s}_u uu^\intercal\right) \geq \beta(f,\alpha)f\left(\sum_{u\in U} \tilde{s}_uuu^\intercal\right),
\end{align*}
 This proves \eqref{eq:goalapplication} as desired.	
\end{proof}

In general, we may not solve \ref{CP:Application} efficiently if $f$ is not a convex function. It turns out that if $f$ is convex and has a certain reciprocal multiplicity property then the integrality gap of the program is small, so assuming further that $f$ is (vector) $k$-monotone, by the above theorem we obtain a composable core-set for the corresponding integral budgeted minimization problems.
In the next two sections we explain two such families of functions namely determinant maximization and optimal design. 
\subsection{Determinant Maximization}\label{sec:detmaximization}
In this section, we use \autoref{thm:mainoptdesign} to prove \autoref{thm:main_upper_bound}.
Throughout this section, for an integer $1\leq k\leq d$ we let $f:\mathbb{S}^+_d\to\Re^+$ be the map $A \mapsto - \det_k(A)^{1/k}$.
 It follows from theory of Hyperbolic polynomials that for any $1\leq k\leq d$, $-\det_k(A)^{1/k}$  is a convex 
function \cite{guler1997hyperbolic}, so one can solve \ref{CP:Application}($-\det_k^{1/k},V,k$) using convex programming.
Furthermore, observe that \ref{CP:Application}($-\det_k^{1/k},V,k$) gives a relaxation of $k$-determinant maximization problem. Nikolov \cite{nikolov2015randomized} showed that any fractional solution can be rounded to an integral solution incurring only a multiplicative error of $e$.
\begin{theorem}[\cite{nikolov2015randomized}]
\label{thm:offlineversion}
There is a  randomized algorithm that for any set $V\subseteq \Re^d$, $1\leq k\leq d$, and any feasible solution $x$  of \ref{CP:Application}($-\det_k^{1/k},V,k$)  outputs  $S \subset V$ of size $k$ such that
\begin{equation}
\label{eq:OfflineGaurantee}
\det_k \left(\sum_{v \in S} vv^\intercal\right) \geq e^{-k}\max_{T\in{V\choose k}} \det_k\left(\sum_{u \in T} uu^\intercal\right)
\end{equation}
\end{theorem}
\begin{proof}We include the proof for the sake of completeness. First, we explain the algorithm:
For $1\leq i\leq k$, choose a vector $v$ with probability 
$\frac{x_v}{k}$ (with replacement) and call it $u_i$.
It follows that,
\begin{equation*}
\mathbb{E}\left[\det_k \left(\sum_{i=1}^k u_iu_i^\intercal\right)\right] = \sum_{S \in {V\choose k}} \frac{k!}{k^k}\Pi_{v\in S} x_v \det_k(\sum_{v \in S} vv^\intercal)  
 \geq e^{-k}    \cdot \sum_{S\in {V\choose k}} \det_k\left(\sum_{v \in S} x_vvv^\intercal\right) 
\end{equation*}
where first equality holds, since we have $k!$ different orderings for selecting a fixed set $S$ of $k$ vectors, but by Cauchy-Binet identity the RHS is equal to $e^{-k}\det_k(\sum_{v\in V} x_vvv^\intercal)$ as desired. 
\end{proof}
Note that the algorithm we discussed in the above proof may have an exponentially small probability of success but \cite{nikolov2015randomized} also gives a de-randomization using the conditional expectation method.
From now on, we do not need convexity.
To use \autoref{thm:mainoptdesign} we need to verify that $- \det_k^{1/k}$ is (vector) $k$-monotone.
\begin{lemma}
For any integer $1\leq k\leq d$, the function $-\det_k^{1/k}$ is vector $k$-monotone.
\end{lemma}
\begin{proof}
Equivalently, we show $-det_k$ is vector $k$-monotone.
Fix $A\succeq 0$, and decompose it as $A=\sum_{a\in A} aa^\intercal$ where we abuse notation and also use $A$ to denote the set of vectors in the decomposition of $A$. Also, fix a vector $v$ and suppose $vv^\intercal \preceq_k B$ for some $B\succeq 0$. We need to show $\det_k(A+vv^\intercal) \leq \det_k(A+B)$. By  \autoref{lem:CauchyBinet},
\begin{eqnarray*}
\det_k (A+vv^\intercal) - \det_k(A) &=&\sum_{S \in {A\choose k-1}} \det_k\left(\sum_{a\in S} aa^\intercal +vv^\intercal\right)\\
&=& \sum_{S\in {A\choose k-1}}\det_{k-1}\left(\sum_{a\in S} aa^\intercal\right) \langle vv^\intercal, \Pi_{\langle S\rangle^\perp}\rangle .
\end{eqnarray*}
The second equality follows by the fact that $\det_k(\sum_{a\in S} aa^\intercal+vv^\intercal)$ is the same as the determinant of the $k\times k$ inner product matrix of all of these $k$ vectors. Using Gram-Schmidt orthogonalization process the latter can be re-written as $\det_{k-1}(\sum_{a\in S}aa^\intercal) \langle vv^\intercal, \Pi_{\langle S\rangle^\perp}\rangle .$
Since $vv^\intercal \preceq_k B$ for any such $S$ we have $\langle vv^\intercal, \Pi_{\langle S\rangle^\perp}\rangle \leq \langle B, \Pi_{\langle S\rangle^\perp}\rangle.$ Therefore,
\begin{eqnarray*}
\det_k (A+vv^T)&=& \det_k(A)+ \sum_{S\in {A\choose k-1}}\det_{k-1}\left(\sum_{a\in S} aa^\intercal\right) \langle vv^\intercal, \Pi_{\langle S\rangle^\perp}\rangle \\
&\leq& \det_k(A) +  \sum_{S\in {A\choose k-1}}\det_{k-1}\left(\sum_{a\in S} aa^\intercal\right) \langle B, \Pi_{\langle S\rangle^\perp}\rangle  \leq \det_k(A+B).
\end{eqnarray*}
The last inequality follows by another application of Cauchy-Binet identity, \autoref{lem:CauchyBinet}.
\end{proof}

Now, we are ready to prove \autoref{thm:main_upper_bound}. Let $V\subseteq \Re^d$ and suppose we are given $p$ subsets $V_1,\dots, V_p$ such that $\bigcup_{i=1}^p V_i=V$.
First, by \autoref{thm:mainoptdesign}, spectral spanners give a $\beta(-\det_k^{1/k},\tilde O(k)) $-composable core-set of size $\tilde O(k)$  for the fractional budgeted minimization problem \ref{CP:Application}($-\det_k^{1/k},V,k)$. Observe that for any $t$,
$$\beta(-\det_k(.)^{1/k},t)=\sup_{A\in\mathbb{S}^+_d}\frac{-\det_k(A)^{1/k}}{-\det_k(tA)^{1/k}} = \sup_{A\in\mathbb{S}^+_d} \frac{\det_k(A)^{1/k}}{t\det_k(A)^{1/k}} =\frac1{t}. $$
So, \autoref{thm:mainoptdesign} gives an $\tilde O(k)$-composable core-set for \ref{CP:Application}($-\det_k^{1/k},V,k)$.
But, by \autoref{thm:offlineversion}, the integrality gap of \ref{CP:Application}($-\det_k^{1/k},V,k$) is $e$. Therefore, \autoref{thm:mainoptdesign}  gives an $\tilde O(k)^k$-composable core-set for integral  determinant maximization problem.
This completes the proof of \autoref{thm:main_upper_bound}.

\subsection{Experimental Design}\label{sec:experimentaldesign}
In this section we discuss another set of applications of \autoref{thm:mainoptdesign} to the problem of {\it 
experimental design}  \cite{pukelsheim1993optimal} 
Consider a noisy linear regression problem: Given $n$ data points 
$v_1,v_2,\dots,v_n \in \Re^d$, we are interested in learning a vector $w \in \Re^d$ from 
observations of the form $y_i=\langle v_i,w\rangle+ \eta_i$ where the noise values $\eta_i$ 
are i.i.d samples from a zero-mean Gaussian distribution. Suppose we are allowed  to learn 
parameter $w$ by only observing  $k\ll n$ data points. Letting $S$ be the set of $k$ chosen 
data points and $\hat{w}$ be the maximum likelihood estimation of $w$, $w-\hat{w}$ has a $d$-
dimensional zero-mean Gaussian distribution with covariance matrix $(\sum_{i \in S} 
v_iv_i^T)^{-1}$. In the {\it experimental design} problem the goal is to choose $k$ data points 
where the corresponding covariance matrix minimizes a given function  $f:\mathbb{S}_d^+ \to 
\mathbb{R}$. The formal definition of the problem is as follows.
\begin{definition}[Experimental Design]
For $V \subset \Re^d$  and $f:\mathbb{S}_d^+\to \Re^+$ and an integer $B$, the experimental design is the problem of finding
\begin{equation*}
S^*(f,k) = \argmin_{S\in {V\choose B}} f\left(\sum_{v\in S} vv^\intercal\right)
,\end{equation*}
where $S$ ranges over all multi-sets of size $B$.
\end{definition}
Experimental design problem has applications to linear bandit \cite{deshpande2012linear,huang2016following}, diversity sampling \cite{kulesza2012determinantal}, active learning \cite{chaudhuri2015convergence}, feature selection and matrix approximation \cite{de2007subset,avron2013faster}, sensor placement in wireless networks \cite{joshi2009sensor}. 

Note that for any function $f$, \ref{CP:Application}($V,f,B$) is a continuous relaxation to the above problem.
It is shown in \cite{singh2018approximation,allen2017near,nikolov2018proportional}
that there is a polynomial time algorithm that if $f$, in addition to being convex and monotone, has a `` reciprocal multiplicity property'', then for $B\gg d$, the solution of \ref{CP:Application}($V,f,B$) can be rounded to an integer solution losing only a constant factor in the value.

We say a function  $f:\mathbb{S}_d^+\to \Re^+$ is {\em regular} if it is convex, monotone and
 $f(tA) =f(A)/t$ for any $t>0$ and $A \succeq 0$.
Here are some examples of regular functions: Average $A\mapsto  \frac{\Tr(A^{-1})}{d}$, Determinant $A \mapsto \det(A)^{-\frac{1}{d}}$,
Min Eigenvalue  $A\mapsto \norm{A^{-1}}_2$, Variance, $A\mapsto \frac{1}{d}\langle \sum_{v\in V} vv^\intercal \Sigma^{-1}\rangle$.
They prove the following.
\begin{theorem}[\cite{allen2017near}]
\label{thm:optdesignoffline}
There exists a polynomial time algorithm that  for any regular function $f:\mathbb{S}_d^+\to \Re^+$, $\eps<1/3$  and $B \geq \frac{5d}{\eps^2}$.  outputs a multi-set $S$ such that
$$f\left(\sum_{v\in S} vv^\intercal\right) \leq (1+8\epsilon) \opt(\ref{CP:Application}(V,f,B)),$$ 
\label{thm:offlineexpdesign}
\end{theorem}
Combining it with  \autoref{thm:mainoptdesign} for $k=d$ leads to the following corollary. 
\begin{corollary}
\label{cor:expdesign}
There exists a polynomial time algorithm which finds an $\tilde O(d)$-composable core-set of size $\tilde O(d)$ for the experimental design with any regular function and $k\geq Cd$ where $C$ is a universal constant.
\end{corollary}
We simply use \autoref{thm:mainoptdesign} and the fact that any regular function is monotone, and hence vector $d$-monotone.
Since for any regular function $f$, $\beta(f,t)=1/t$, we obtain an $\tilde O(d)$-composable core-set of size $\tilde O(d)$ for the fractional version of the experimental design problem. 
But then, by \autoref{thm:offlineexpdesign} any $\alpha$-composable core-set for the fractional experimental design problem is an $O(\alpha)$-composable core-set for (integer) experimental design.
%
Again, we emphasize that given $V,B$, for any regular function our algorithm outputs exactly the same composable core-set.

In Section \ref{sec:lower_bound} we show that for many examples of regular functions $f$, the above bound is almost optimal.

\section{Lower Bound}\label{sec:lower_bound}
In this section, we study lower-bounds on the 
approximation ratio and size of the composable core-sets 
for the $k$-determinant maximization  and the experimental 
design problem. In particular, we prove 
\autoref{thm:main_lower_bound}.
 We also prove the bound given by 
\autoref{cor:expdesign} is optimal up to a logarithmic 
factor for some of the regular functions. 
\subsection{Construction of a Hard Input}
Here, we describe a distribution over collection of vectors which turns out to be  a``hard'' input for composable core-sets in many spectral problems. We use that in the next subsection to establish our lower-bound results. The construction of the instance is described in \ref{alg:BadExample}.
\begin{figure*}
\begin{mdframed}[frametitle=\setword{Figure 1}{alg:BadExample}: A Hard Input for Composable Core-sets,  frametitlerule=true,    frametitlerulewidth=1pt]
\begin{enumerate}[wide, labelwidth=!, labelindent=0pt]
\item Set $m=\frac{d}{\log d}$ to have  $d^{m/d}=O(1)$.
\item Consider a set $G \subset \Re^{m+1} $ of $n=d^{\beta+2}$ 
vectors such that for every two vectors $p,q\in G$, we have 
\begin{equation}
\label{eq:randomvectors}
\langle p,q\rangle \leq O\left(\frac{\sqrt{\beta}\log d}{\sqrt{d}}\right)
\end{equation}
\item Do the following for any $1 \leq i \leq d-m$: 

Embed $G$ into the space spanned by $e_1,\ldots,e_m$ and 
    $e_{m+i}$, and call it $G_i$. 
Choose an index $\pi_i \in [n]$ uniformly at random. Construct $X_i$ by rotating $G_i$ using a rotation $R(\pi_i) \colon \Re^{d} \to \Re^{d}$ that maps the $\pi_i$-th vector in $G_i$ to $e_{m+i}$, and that maps the rest of the vectors in $G_i$ to points in $\langle e_1,\ldots ,e_m, {\color{red}e_{m+i}}\rangle$. 
\item Choose a matrix $Q$ uniformly at random from the Haar measure over the space of rotations in $\Re^{d}$ (i.e., orthogonal $d \times d$ matrices). 
\item Return $QX_1,\ldots,QX_{d-m}$ and $QY_1, \ldots, QY_m$ where $Y_i= Me_i$ for a large enough scalar $M$.
\end{enumerate}
\end{mdframed}
\end{figure*}


To construct the instance we need to start with a set of vectors $G$ satisfying \eqref{eq:randomvectors}. 
The following lemma guarantees the existence of the set $G$.
\begin{lemma}[Implied by \cite{DG-AEPJL-99}]\label{lem:orth}
Let $G$  be a set of $d^{\beta+2}$ vectors chosen independently and uniformly at random from the $(m+1)$-dimensional unit sphere for $m=\frac{d}{\log d}$ and $\beta \geq 1$. Then with  with probability at least $1-1/d^3$,  for every two vectors $p,q\in G$, we have $\langle p,q\rangle \leq O(\frac{\sqrt{\beta} \log d}{\sqrt d})$.
\end{lemma}
\begin{proof}
Let $\eps = \frac{C_1\sqrt{\beta}\log d}{\sqrt d}$ for some constant $C_1$ that we specify later. For any two random vectors chosen uniformly at random from the $(m+1)$-dimensional unit sphere, their inner product is distributed as $\mathcal{N}(0,1)/\sqrt m$. Using Lemma 2.2 (b) from \cite{DG-AEPJL-99}, the probability that their inner product is more than $\eps$, is bounded by $e^{-\frac{\eps^2}{3}\cdot m}$ (note that this uses the fact that $\eps^2m>27$).

Thus, by union bound, the probability that for any pair of points in $G$, their inner product is bounded by $\eps$, is at least $1 - d^{2\beta+4} \cdot e^{-\frac{\eps^2}{3}\cdot m} \geq 1- d^{2\beta+4} \cdot e^{-\frac{\beta(C_1\log d)^2}{3d}\cdot \frac{d}{\log d}} \geq 1 - d^{-\beta C_1^2/3 + 2\beta +4}$. Setting $C_1=6$, this probability is at least $1 -1/d^3$.
\end{proof}

The main property of the sets generated in \ref{alg:BadExample} that we use is the following.
\begin{lemma}\label{l:miss} Let $c$ be an arbitrary core-set function. For any $i=1 \ldots d-m$, the probability that the image of $e_{m+i}$ under $Q$ is in the core-set for $QX_i$ is bounded by $\frac{\abs{c(X_i)}}{\abs{X_i}}$, i.e.,
\[ \Pr_{Q,\pi} [  Qe_{m+i} \in c( Q R(\pi_i) G_i ) ] \le \frac{\abs{c(X_i)}}{\abs{X_i}}\]
\end{lemma}
\begin{proof}
From the right-translation-invariance of Haar measure, it follows that, for any fixed value of $\pi_i$, the distribution of $Q R(\pi_i) G_i$ is the same as the distribution of $Q G_i$. 
Therefore,  the joint distribution of $(\pi_i, Q R(\pi_i) G_i)$ is the same as of $(\pi_i, QG_i)$, so it suffices to bound  $\Pr_{Q,\pi} [  (G_i)_{\pi_i} \in c( Q G_i )  ]$, where  $(G_i)_{\pi_i}$ denotes the $\pi_i$-th vector in $G_i$. 
Since  $\pi_i$ and $QG_i$ are independent, the bound follows.  
\end{proof}

\subsection{Lower-bounds for Composable Core-sets for Spectral Problems}
Consider the collection of sets generated by the procedure described in \ref{alg:BadExample}.
Without loss of generality we may assume $Q=I$, as  rotation matrices do not change spectral quantities  we are interested in. So let $X_1,\ldots,X_{d-m}$ and $Y_1,\ldots,Y_m$ be the output sets. 
We are only interested in polynomial size core-sets, so fix
 a core-set function $c$ that  maps any set in $\Re^d$ to its 
subsets of size at most $d^\beta$ for some constant $\beta \geq 1$. 
Using  Lemma~\ref{l:miss} and union bound, the probability that for at least one $1 \leq i \leq d-m$ we have $e_{m+i}\in c(X_i)$, is at most  $(d-m)\cdot \frac{\abs{c(X_i)}}{\abs{X_i}}\leq \frac{(d-m)d^{\beta}}{d^{\beta+2}}\leq 1/d$. So WLOG we can assume $(Q)e_{m+i}=e_{m+i} \notin c(X_i)$ for any $1 \leq i \leq d-m$. It implies the following assumption that   
we crucially use in the future proofs.

\noindent\textbf{Assumption.} For any $u \in \bigcup_{i=1}^{d-m} c(X_i)$, 
\begin{equation}
\label{eq:instanceprop}
\left \langle \Pi_{\langle e_{m+1},\ldots,e_d \rangle},uu^\intercal  \right \rangle = \sum_{j=1}^{d-m} \langle u,e_{m+j}\rangle^2 \leq O\left(\frac{\beta \log^2 d}{d}\right).
\end{equation}
To see this, suppose $u \in c(X_i)$ for some $i$. Since $X_i \subset \langle e_1,\dots,e_m,e_{m+i}\rangle$ by construction, we have $\langle u,e_{m+j}\rangle =0$ for $j\neq i$. Moreover, we assumed $e_{m+i}\notin c(X_i)$, so $\langle u,e_{m+i} \rangle \leq O(\frac{\sqrt{\beta}\log d}{\sqrt{d}})$ by \eqref{eq:randomvectors}.  

We also define  
\begin{equation}
\begin{array}{ccc}
\label{def:collections}
V=\left(\bigcup_{i=1}^{d-
m}X_i\right) \bigcup \left(\bigcup_{j=1}^m Y_j\right)& \text{and} &U=\left(\bigcup_{i=1}^{d-m}c(X_i)\right)\bigcup\left(\bigcup_{j=1}^m c(Y_j)\right).
\end{array}
\end{equation}
In what follows we assume \eqref{eq:instanceprop} holds.

\begin{proofof}{\autoref{thm:main_lower_bound}}
First let us proof the theorem for $k=d$. Consider the core-set function $c$ and input sets $X_1,\dots,X_{d-m},Y_1,\dots,Y_m$ explained above. 
Consider the optimal set of $d$ vectors maximizing the determinant on the union of the input sets, $V$. It 
contains $ e_{m+1},\cdots, e_{d}$ from the sets $X_1,\cdots, X_{d-m}$, 
respectively, and the points $M  e_1,\cdots,M  e_m$ from the sets $ Y_1,\cdots, 
 Y_m$ respectively.  The value of this solution is equal to $(M^m)^2$. At the same 
time, the optimal solution from the union of the core-sets $U$ must contain the $m$ 
vectors $M  e_1,\cdots,M e_m$ from the sets $Y_1,\cdots,Y_m$, if we set $M$ to 
be large enough. Any other set of $k-m=d-m$ vectors must be chosen from the union 
of core-sets $c(X_i)$. So by Hadamard inequality we get  the optimum is at most $(M^m)^2\cdot 
\max_{u \in U} \left(\langle\Pi_{\langle e_{m+1},\ldots,e_d \rangle^\perp},uu^\intercal 
\rangle\right)^{d-m}$ which results in a value of at most $(\frac{M^m}{(\sqrt d/ 
(O(\sqrt{\beta}) \log d))^{d-m}})^2 = \frac{M^{2m}(O(\sqrt{\beta}) \log d)^{2(d-m)}}{d^{d-m}}$ 
by  assumption \eqref{eq:instanceprop}. 
 Hence the approximation ratio is at least  $(d/(O(\sqrt{\beta})\log d)^2)^{d-m}$. Noting $m=\frac{d}{\log d}=o(d)$ completes the proof for $k=d$. 

To extend the above proof for smaller $k$, we can consider the same instance in $d' = k$ dimensions, and then append the vectors with $d-d'$ zeros. It is straight-forward to see this gives us the same result for any value of $k\leq d$, yielding Theorem~\ref{thm:main_lower_bound}.
\end{proofof}

Now, we present our lower-bounds on the approximation ratio of composable core-sets for the experimental design problem. Again, we consider the aforementioned core-set function $c$ and input sets $X_1,\dots,X_{d-m},Y_1,\dots,Y_m$.   
\begin{proposition}
\label{prop:lower_boundexp1}
Composable core-sets of size at most  $d^{\beta}$ for the experimental design problem with respect to the function $A\mapsto \norm{A^{-1}}_2$ and size parameter $B \geq Cd$ have an approximation factor of at least $O(\frac{d}{\beta \log^2 d})$, for any $\beta \ge 1$ and a universal constant $C$.
\end{proposition}
\begin{proof}
Note that it is enough to show the same lower-bound for the corresponding fractional budget minimization problem (\ref{CP:Application}). Since as pointed out in \autoref{sec:app}, the relaxation \ref{CP:Application} has constant integrality  gap when the conditions of \autoref{thm:offlineexpdesign} are satisfied (which is satisfied for large enough $C$).
Therefore, we show for any $B$ and $f=(A \mapsto \norm{A^{-1}}_2)$,
\begin{equation*}
\frac{\opt(\text{\ref{CP:Application}}(U,f,B))}{\opt(\text{\ref{CP:Application}}(V,f,B))} \geq \Omega\left(\frac{d}{\beta\log^2 d}\right)
\end{equation*}
where $V$ and $U$ are defined by \eqref{def:collections}.
Let us first find an upper bound on the optimal on $V$. For simplicity we work with the reciprocal of $f$ (note that for any $A \in \mathbb{S}_d^+$, $\frac1{f(A)}= \lambda_d(A)$).  Picking  $Me_i \in B_i$ for any $1 \leq i \leq m$, and $e_{m+i} \in A_i$ for any $1 \leq i \leq d-m$, all with multiplicity $\frac{B}{d}$, we can deduce  
$$ \frac{1}{\opt (\text{\ref{CP:Application}}(V,f,B))} \geq  \lambda_d\left(\frac{B}{d}\cdot \left(\sum_{i=1}^{d-m} e_{m+i}e_{m+i}^\intercal +\sum_{j=1}^m M^2e_je_j^\intercal\right)\right) \geq \frac{B}{d},$$
for $M>1$. So in order to prove the theorem, is suffices to show for any feasible solution $s \in \Re^U$ of \ref{CP:Application}$(U,f,B)$ (which means  $\sum_{u \in U} s_u\leq B$), we have $\lambda_d\left(\sum_{u \in U} s_u uu^\intercal\right)\leq O\left(\frac{\beta \cdot B\log^2 d}{d^2}\right)$. We have 
\begin{align*}
\lambda_d\left(\sum_{u \in U} s_uuu^\intercal\right) &\leq \frac1{d-m}\sum_{i=m+1}^{d} \lambda_i\left(\sum_{u \in U} s_uuu^\intercal\right) &\\
&\leq \frac1{d-m}\left\langle \sum_{u \in U} s_uuu^\intercal ,\Pi_{\langle e_{m+1},\ldots,e_d \rangle}  \right\rangle  & \text{By \autoref{lem:minquadform}}\\
&\leq \frac{\sum_{u \in U} s_u}{d-m} \cdot O\left(\frac{\beta \log^2 d}{d}\right) & \text{By  \eqref{eq:instanceprop}}
\end{align*}
which  completes  the proof, as $\sum_{u \in U} s_u\leq B$.
\end{proof}
\begin{proposition}
\label{prop:lower_boundexp2}
Composable core-sets of size at most  $d^{\beta}$ for the experimental design problem with respect to the function $A\mapsto \det(A)^{-1/d}$ and size parameter $B \geq Cd$ have an approximation factor of at least $O(\frac{d}{\beta \log^2 d})$, for any $\beta \ge 1$ and a universal constant $C$.
\end{proposition}

\begin{proof}
Similar to the previous proposition, it is enough  to show that for function $f=A\mapsto \det\vspace{0cm}^{-1/d}$ and  any $B$   
\begin{equation*}
\frac{\opt(\text{\ref{CP:Application}}(U,f,B))}{\opt(\text{\ref{CP:Application}}(V,f,B))} \geq \Omega\left(\frac{d}{\beta\log^2 d}\right),
\end{equation*}
where $U$ and $V$ are defined in \eqref{def:collections}.
let us first find  an upper bound  
on the optimum on $V$ (or equivalently a lower-bound on its reciprocal). If we choose $e_{m+i}$ from $X_i$ for any $1 \leq i 
\leq d-m$ and $Me_j$ from $Y_j$ for any $1 \leq j \leq m$ with equal weights of $B/d$, we get $\frac{1}{\opt\text{\ref{CP:Application}}(V,f,B)} \geq  \frac{B}{d}M^{2m/d}.$ 
So in order to prove the theorem it suffices to show $$\frac{1}{\opt(\text{\ref{CP:Application}}(U,f,B))}=\det\left(\sum_{u \in U} 
s_uuu^\intercal\right)^{1/d} \leq BM^{2m/d}\cdot O\left(\frac{\beta \log^2 d}{d^2}\right)$$ 
for any feasible solution $s\in \Re^U$, i.e. $\sum_{u \in U} s(u)\leq B$. By Cauchy-Binet we know $\det\left(\sum_{u \in U} s_uuu^\intercal\right)=\left(\sum_{S \in \binom{U}{d} } \det(\sum_{u \in S} s_uuu^\intercal) \right)$. Taking $M$ to be large enough implies the summation is dominated by terms containing all vectors $Me_1,\dots,Me_m$. So letting $H=\langle e_{m+1},\dots,e_d \rangle$, we have   
\begin{equation*}
\begin{aligned}
\label{eq:upperboundopt}
\det\left(\sum_{u \in U} s_uuu^\intercal\right)&=
M^{2m}(1+o(1))\cdot  \sum_{\substack{S\in \binom{U}{d-m}}} \det\left(\sum_{j=1}^m s_{e_j}e_je_j^\intercal+  \sum_{u \in S} s_{u} uu^{\intercal}\right)  \\ 
&=
M^{2m}(1+o(1))\cdot  \prod_{j=1}^m s_{e_j}\cdot \sum_{\substack{S\in \binom{U}{d-m}}} \det_{d-m}\left(\sum_{u \in S} s_{u} \Pi_{H }(u_i)\Pi_{H}(u_i)^\intercal \right) \\ 
\end{aligned}
\end{equation*}
Now, note that if $p,q \in U$ both belong to the same  $c(X_i)$, then the corresponding determinant in the summation is zero as $p,q \in \langle e_1,\dots,e_m,e_{m+i} \rangle$ by construction. So we have   
\begin{equation}
\begin{aligned}
\label{eq:lastupperbound}
\det\left(\sum_{u \in U} s_uuu^\intercal\right)&=  M^{2m}(1+o(1))\cdot \prod_{j=1}^m s_{e_j} \cdot \sum_{\substack{(u_1,\dots,u_{d-m}) \\ \in \\ c(X_1)\times \dots,\times c(X_{d-m})}} \det_{d-m}\left(  \sum_{i=1}^{d-m} s_{u_i}\Pi_{H }(u_i)\Pi_{H}(u_i)^\intercal  \right) \\
&\leq M^{2m}(1+o(1)) \cdot\prod_{j=1}^m s_{e_j} \cdot \left(\prod_{i=1}^{d-m}\sum_{u \in c(X_i)} s_u\right) \cdot  \max_{S \in \binom{U}{d-m}}\det_{d-m}\left( \sum_{u\in S} \Pi_{H}(u)\Pi_{H}(u)^{\intercal}\right)  
\end{aligned}
\end{equation}
We can further simplify the above by  combining Hadamard inequality and \eqref{eq:instanceprop}. It implies 
\begin{equation}
\label{eq:HadamardInequalityapp}
\max_{S \in \binom{U}{d-m}}\det_{d-m} \left(\sum_{u\in S} \Pi_{H}(u)\Pi_{H}(u)^\intercal\right) \leq  \max_{S \in \binom{U}{d-m}}  \prod_{u \in S} \norm{\Pi_{H}(u)}^2 \leq O\left(\frac{\beta \log^2 d}{d}\right)^{(d-m)} 
\end{equation}
Furthermore, by AM-GM inequality we get that $\prod_{j=1}^m s_{e_j} \cdot \left(\prod_{i=1}^{d-m}\sum_{u \in c(X_i)} s_u\right)\leq 
\left(\frac{\sum_{u \in U} s_u}{d}\right)^d\leq \frac{B^d}{d^d}$. 
Combining the above with \eqref{eq:lastupperbound} and \eqref{eq:HadamardInequalityapp} proves 
$$ \det\left(\sum_{u \in U} s_uuu^\intercal\right)^{1/d} \leq M^{2m/d}\cdot \frac{B}{d}\cdot  O\left(\frac{\beta \log^2 d}{d}\right)^{(d-m)/d}.$$
Noting $d^{2m/d}=O(1)$ completes the proof.
\end{proof}

 \bibliographystyle{alpha}
 
 \bibliography{biblio}
 
\end{document}